%% file: main.tex
\documentclass[conference]{IEEEtran}
\IEEEoverridecommandlockouts
\usepackage{cite}
\usepackage{amsmath,amssymb,amsfonts}
\usepackage{graphicx}
\usepackage{textcomp}
\usepackage{xcolor}

\usepackage{graphicx}
\usepackage{balance}  

\usepackage{color}
\usepackage{multirow}
\usepackage{subfigure}
\usepackage[export]{adjustbox}

\newcommand{\stitle}[1]{\vspace{0.8ex}\noindent\textup{\textbf{#1}}}

\usepackage{algorithmicx}
\usepackage{algorithm}
\usepackage{algpseudocode}

\newtheorem{definition}{Definition}
\newtheorem{proof}{Proof}
\newtheorem{lemma}{Lemma}

\newcommand{\hide}[1]{}

\def\BibTeX{{\rm B\kern-.05em{\sc i\kern-.025em b}\kern-.08em
    T\kern-.1667em\lower.7ex\hbox{E}\kern-.125emX}}
\begin{document}

\title{Spatio-temporal flow patterns}


\author{
\IEEEauthorblockN{Chrysanthi Kosyfaki\IEEEauthorrefmark{1}, Nikos
  Mamoulis\IEEEauthorrefmark{2}, Reynold Cheng\IEEEauthorrefmark{1}, Ben Kao\IEEEauthorrefmark{1}}
\IEEEauthorblockA{\IEEEauthorrefmark{1}Department of Computer Science,
  University of Hong Kong,  \{kosyfaki,ckcheng,kao\}@cs.hku.hk}
\IEEEauthorblockA{\IEEEauthorrefmark{2}Dept. of Computer Science
  and Engineering, University of Ioannina, nikos@cs.uoi.gr}
}

\maketitle

\begin{abstract}
Transportation companies and organizations routinely collect huge
volumes of passenger transportation data.
By aggregating these data (e.g., counting the number of passengers
going from a place to another in every 30 minute interval), it becomes possible to 
analyze the movement behavior of passengers in a metropolitan area. 
In this paper, we study the problem of finding important trends in
passenger movements at varying granularities, which is useful in a wide
range of applications
such as target marketing, scheduling, and travel intent prediction.
Specifically, we study the extraction of movement patterns between
regions that have significant flow.
The huge number of possible patterns render their detection 
computationally hard. We propose algorithms that greatly reduce the search space and
the computational cost of pattern detection.
We study variants of patterns that could be useful to different
problem instances, such as constrained patterns and top-$k$ ranked patterns.

\end{abstract}


\input{introduction}

\input{relatedwork}

\input{definitions}
\input{algorithms}

\input{extensions}
\input{experiments}

\section{Conclusions}\label{sec:conclusion}
In this paper we have studied the problem of enumerating
origin-destination-timeslot (ODT) patterns of varying granularity from a
database of trips.
To our knowledge, this is the first work that formulates and studies
this problem. Due to the huge number of region-time combinations that
can formulate a candidate pattern, the problem is hard. We explore
the problem space level-by-level, building on a weak monotonicity
property of patterns. We propose a number of optimizations that
greatly reduce the cost of the baseline pattern enumeration algorithm.
To reduce the possibly huge number of ODT patterns, which take too
long to enumerate and analyze, we propose practical variants of the
mining problem, where we restrict the size of patterns
and/or the region/timeslots included in them.
In addition, we suggest the interesting definition of rank-based
patterns and we study their efficient enumeration.
Experiments with three real datasets demonstrate the effectiveness of
the proposed techniques. In the future, we plan to study the
relationships between patterns at different levels/granularity and
alternative definitions of interesting ODT patterns.

\bibliographystyle{abbrv}
\bibliography{references}

\end{document}

%% file: introduction.tex
\section{Introduction}\label{sec:intro}
Consider a transportation company such as a metro system, which
routinely collects large volumes of data from its passengers,
regarding their entrance and exit points in the system and the times
of their trips.
Information of individual trips can be used in personalized services,
after obtaining consent from the passengers. Other than that, it is
hard to use such detailed data, mainly due to
privacy constraints.
On the other hand, aggregate information about passenger
trips can be valuable to the company, since it can provide estimates
and predictions about the
passenger flow between regions 
at different times of the
day and different days of the week. 

\stitle{Our contribution} In this paper, we study the problem of identifying interesting
{\em origin-destination-time} patterns of passengers, called ODT-patterns
for brevity,
at varying granularity. 
For this, we first use the application domain to define the finest
granularity of regions on the map
(e.g., each region corresponds to a
metro station) and also define the finest time
intervals of interest (e.g., divide the 24-hour time interval of a day into 48
30-minute timeslots).
We call these {\em atomic} regions and {\em
  atomic} timeslots, respectively.

Since the durations of all trips from a given
origin to a given destination at a given time are strongly correlated,
the time of reaching a destination can be inferred from the time when
the trip starts.
Hence, a trip can be described by its origin region, its destination
region, and the timeslot when the trip starts, i.e., as an $(o,d,t)$ triple. 
For all $(o,d,t)$ triples, where $o$ and $d$ are (different) atomic
regions and $t$ is an atomic timeslot, we measure the total number of
passengers who took a trip from $o$ to $d$ at time $t$.
The total flow of an $(o,d,t)$ triple
characterizes its importance; the triples with high flow are
considered to be important and they are called {\em atomic
  ODT-patterns} (we drop the ODT prefix whenever the context is clear).
In a {\em generalized} ODT-triple, denoted by $(O,D,T)$, $O$ and $D$
are sets of neighboring
atomic regions and $T$ consists of one or more consecutive timeslots.
An atomic $(o,d,t)$ triple is a component of an  $(O,D,T)$ triple if
$o\in O$, $d\in D$, and $t\in T$. 
$(O,D,T)$ is {\em non-atomic}, if it has more than one components, i.e.,
at least one of $O$, $D$, or $T$ is non-atomic.

Defining and finding important non-atomic patterns is more
challenging. One reason is that the number of possible atomic region
combinations that can form a generalized (i.e., non-atomic) region $O$
or $D$ is huge and it is not practical to consider all these
combinations and their flows.
At the same time, for a given generalized
ODT triple, it is hard to estimate the flow quantity that can be deemed
significant enough to characterize the triple an interesting
pattern.
To solve these issues, we follow a ``voting'' approach, where we
characterize an ODT triple as a pattern if at least a certain percentage
of its constituent $(o,d,t)$ triples are atomic patterns (i.e., they
have large enough flow). This allows us to design and use a
pattern enumeration algorithm,
which, starting from the atomic patterns,
identifies all ODT patterns progressively by synthesizing them
from less generalized ODT patterns.
We propose a number of optimizations to our algorithm, which
significantly reduce the time spent for generating candidate patterns
and counting their supports.

Despite our optimizations, ODT pattern enumeration can
still be expensive due to the essentially huge number of generated and counted
patterns even with relatively high support thresholds for atomic patterns. 
Given this,
we also study practical variants of ODT pattern
search.
We investigate the detection of patterns which are {\em constrained} to a
subset of regions and timeslots, which reduces the problem size and renders
pattern enumeration much faster.
Besides, this
allow us to define the
importance of flow
parametrically in a {\em fair} manner (i.e., by constraining pattern search
to under-represented regions).
We also study pattern detection by limiting the number of
atomic regions and timeslots that a pattern may have.
Finally, we define and solve the problem of finding the
top-ranked patterns at each granularity level.
We propose an efficient algorithm that outperforms the baseline approach
of finding all patterns at each level and then selecting the top ones
by a wide margin.

\stitle{Applications} Identifying spatio-temporal flow patterns finds
several applications, e.g., in transportation networks
\cite{gudmundsson2004efficient,
  DBLP:journals/tkde/LiuLJXDTZ23,DBLP:conf/kdd/WangYCW0019}, weather forecasting \cite{takafuji2020spatiotemporal, heng2020spatiotemporal},
social networks, etc \cite{DBLP:conf/adma/CunhaSR14,liu2021exploring}.
In transportation networks, 
detection of passenger movement patterns
can facilitate 
the handling of emergencies or incidents.
For instance in December 2021, there
was an accident in Hong Kong subway system.%
\footnote{https://www.thestandard.com.hk/breaking-news/section/4/183861/(Video)-MTR-door-flew-off,-disrupting-peak-hour-service}
As a result, scheduled trips were canceled and passengers had to be
served by other means (i.e., buses).
Spatio-temporal flow patterns could help in predicting the movement
needs and for scheduling on-demand transportation
for affected passengers. 
As another application, studying
the evolution of patterns can help in scheduling future trips more effectively.
Patterns can also help to understand the correlations between map
districts and perform target-marketing, cross-district advertisements,
or location planning.



\stitle{Outline} 
Section \ref{sec:relwork} reviews related work on spatio-temporal
pattern mining.
In Section \ref{sec:def}, we formally
define the problem we study in this paper. Section \ref{sec:algo}
presents an algorithm for extracting spatio-temporal flow patterns and
its optimizations.
In Section \ref{sec:ext}, we define interesting variants of
flow patterns and propose algorithms for their enumeration.
Section \ref{sec:exps} evaluates our
methods on real networks with different characteristics. Finally,
Section \ref{sec:conclusion} concludes the paper with a discussion
about future work.

%% file: relatedwork.tex
\section{Related Work}\label{sec:relwork}
Spatio-temporal patterns are spatial events, correlations, or sequences
(trajectories) that repeat themselves over time.
Spatio-temporal pattern mining is a well-studied problem in the
literature
\cite{DBLP:conf/sdm/GiannottiNP06,DBLP:journals/tits/AsifDGOFXDMJ14,DBLP:conf/kdd/Morimoto01,DBLP:conf/kdd/ZhangMCS04,DBLP:conf/icsdm/YooB11,DBLP:journals/eswa/Yu16,DBLP:journals/datamine/HanPYM04,DBLP:conf/edbt/KosyfakiMPT19,DBLP:journals/gis/CaiK22,DBLP:conf/icde/KosyfakiMPT21,DBLP:journals/air/AnsariAKBM20,DBLP:conf/vldb/HadjieleftheriouKBT05,zhang2018predicting, DBLP:conf/wsdm/ParanjapeBL17,la2002spatio,tan2001finding},
where a number of different problem definitions and solutions are presented.

Agrawal and Srikant \cite{DBLP:conf/icde/AgrawalS95} introduced the
concept of sequential patterns over a database of customer
sales transactions. More specifically, the problem of
mining sequential  patterns is to find the maximal sequences of
itemsets (that appear together in a customer transaction) among all
those that have a certain user-specific minimum support. The
authors use three different algorithms to solve this problem and
evaluate their proposed techniques using synthetic data.

Extracting trajectory patterns from large graphs is a well-studied
problem in data mining
\cite{DBLP:conf/sigmod/NgC04,DBLP:reference/gis/GudmundssonLW08,
  DBLP:journals/jips/KangY10, DBLP:conf/ssd/KalnisMB05,DBLP:conf/sigmod/TaoFPL04}.
The main objective is to
find spatio-temporal patterns from raw GPS data, which can describe,
for example, frequent routes or passenger movements.
Giannotti et al. \cite{DBLP:conf/kdd/GiannottiNPP07}
extended the problem of mining sequential patterns in
trajectories. They define trajectory patterns as frequent behaviors in
both space and time.
They also propose algorithms for discovering regions of
interest,
to mine trajectory patterns with predefined regions
and reduce the complexity of the problem.
Cao et al. \cite{DBLP:conf/icdm/CaoMC05} studied the
problem of mining sequential patterns from spatiotemporal
data. They defined patterns as sequences of spatial regions,
they define regions using clustering, and then identify
sequential patterns of regions that repeat themselves over time.
They also
proposed a substring tree, a fast approach for extracting longer
patterns. 
Choi et al., \cite{DBLP:journals/pvldb/ChoiPH17} introduce a tool
for discovering all regional movement patterns in semantic
trajectories.
They design an algorithm called RegMiner (Regional Semantic Trajectory
Pattern Miner)
which is capable of finding movement patterns that can be frequent
only in
specific regions and not in the entire space. By doing this, they
automatically
reduce the search requirements and identify more interesting patterns.
Fan et al. \cite{DBLP:journals/pvldb/FanZWT16}
provide scalable trajectory mining methods using Apache Spark.
Pattern mining in graph streams has been studied in
\cite{DBLP:journals/pvldb/AggarwalLYJ10}, where the authors propose 
probabilistic algorithms.
The goal is to develop a summarization of the graph stream which can
then be used as input to the mining problem.
They use a min-hash approach for extracting patterns 
efficiently.


Our problem is quite different compared to previous work on
trajectory, sequence, and graph mining.
First, we are not interested in finding frequent paths
(subsequences, subgraphs), but in finding hot combinations of trip origins,
destinations, and timeslots. Second, we do not search for patterns at
the finest granularity only, but looking for patterns where any of
the three ODT components are generalized. Furthermore, we only have
a weak monotonicity property when generalizing detailed patterns,
which means that the classic Apriori algorithm (and its variants) \cite{DBLP:conf/sigmod/AgrawalIS93,DBLP:journals/tkde/HanF99,DBLP:conf/ssd/KoperskiH95,DBLP:books/mit/fayyadPSU96/AgrawalMSTV96,DBLP:reference/gis/X08bq}
cannot be readily applied to solve our problem.

Mining traffic flow patterns is a problem that recently attracted a
lot of attention  
\cite{DBLP:journals/tkde/GongLZLZ22,
  DBLP:journals/tkde/LiuLJXDTZ23,
DBLP:conf/kdd/WangYCW0019,
DBLP:journals/tits/DuPWBWGLL20,xie2020urban}.
Liu et al. \cite{DBLP:journals/tkde/LiuLJXDTZ23}  studied the
problem
of extracting traffic flow knowledge
from transportation data, i.e., pairs of POIs on the map that have
significant traffic flow between them.
Wang et al.,
\cite{DBLP:conf/kdd/WangYCW0019} develop a model for predicting the
flow density in different regions. To do this, they represent a city as
a grid and take into consideration regions that may have a significant
amount of flow compared to other less important regions. Although we
also measure flow between regions, we are not interested
in predicting the flow distribution. Moreover, none of these previous
works studies flow patterns at {\em different granularities}, where regions
and time periods may consist of multiple atomic elements; hence,
previous works cannot be applied to solve our problem. 



%% file: definitions.tex
\section{Definitions}\label{sec:def}
In this section, we formally define ODT patterns and the graph
wherein they are identified.
Table \ref{tab:notations} shows the
notations used frequently in the paper.

\begin{table}[ht]
\caption{Table of notations}\label{tab:notations}
\centering
\footnotesize
\begin{tabular}{|c|c|}
\hline
Notations & Description \\
\hline  
$G(V,E)$ & region neighborhood graph  \\
$r_i$ & atomic region \\
$R_i$ & region \\
$t_i$ & atomic timeslot \\
$T_i$ & timeslot\\
$P$ & atomic ODT pattern or triple \\
$P.O$ & pattern/triple origin \\
$D$ &  pattern/triple destination \\
$T$ &  pattern/triple timeslot \\
$\sigma(P)$ & support of atomic ODT pattern $P$ \\
$P$.cnt & number of atomic patterns in ODT pattern $P$ \\
$\mathcal{P_\ell}$/$\mathcal{T_\ell}$ & Set of ODT patterns/triples at
  level $\ell$\\
\hline         
\end{tabular}
\end{table}

The main input to our problem is a {\em trips} table,
which
records information about trips from origins to destinations at
different times. Each origin/destination is a 
minimal region of interest on a map (e.g., a
district, a metro station, etc.), called {\em atomic region}.
Let $V$ be the set of all atomic regions.
An undirected neighborhood graph
$G(V,E)$ defines the neighboring relations between atomic regions;
there is an edge
$(v,u)$ in $E$ iff $v\in V$ and $u\in V$ are neighbors on the
map.
Finally, the timeline is divided into periods that repeat themselves (e.g., 24-hours each) 
and each period is discretized into time ranges (e.g., 48 30-minute slots).
Each such minimal time range is called {\em atomic} timeslot.
Periods can also be classified (e.g., weekdays vs. weekends); hence
atomic timeslots may refer to different period classes (e.g.,
8:00-8:30 on weekdays).
Figure \ref{fig:region} shows an exemplary region neighborhood graph with four atomic regions (districts or stations) as vertices and Figure
\ref{fig:table} shows a (snapshot of a) trips table,
which includes individual trips that have taken place between these regions.

\begin{figure}[h]
	\centering
	\subfigure[Region graph]{
		\label{fig:region}
		\includegraphics[width=0.27\columnwidth]{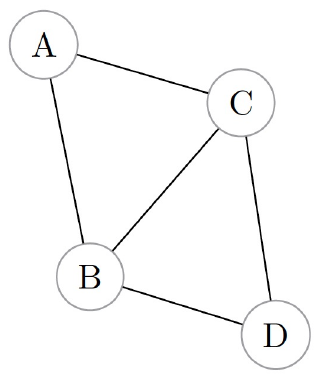}}
	\hspace{4mm}
	\subfigure[Trips table]{
	\label{fig:table}
	\includegraphics[width=0.5\columnwidth]{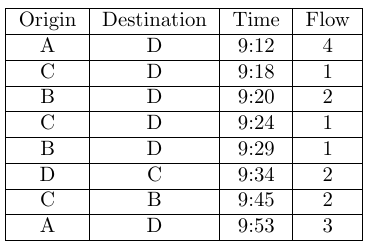}}
	\caption{Example of input data}
	\label{fig:example_def}
\end{figure}


  \begin{definition} [Region/Timeslot] \label{def:eregion}
 A  region $r$ is a subset $V'$ of $V$, such that the
induced subgraph $G'(V',E')$ of $G$ is connected. A
timeslot $T$ is a continuous sequence of atomic timeslots.
  \end{definition}



\begin{definition} [Generalization of a region/timeslot] \label{def:genrt}
 A region $R_1$ is a generalization of region $R_2$ iff
 $R_2\subset R_1$. A timeslot $T_1$ is a 
generalization
of timeslot $T_2$ iff
$T_2\subset T_1$.
\end{definition}

\begin{definition} [Minimal generalization of an region/timeslot] \label{def:mingenrt}
 A region $R_1$ is a minimal generalization of region $R_2$ iff
 $R_2\subset R_1$
and  $R_1 - R_2$ is an atomic region. A timeslot $T_1$ is a minimal
generalization
of timeslot $T_2$ iff
$T_2\subset T_1$ and $T_1 - T_2$ is an atomic timeslot.
\end{definition}

For example, region $\{B,C,D\}$ is a minimal generalization of $\{B,D\}$.
Symmetrically,
$\{B,D\}$ is a  minimal specialization of $\{B,C,D\}$.

We found that the start and end times of individual trips between the same origin
and destination are strongly correlated.
Specifically, we computed the mean
absolute deviation
(MAD) of trip durations
between all origin-destination pairs and for all timeslots of the
origin time
in both of our real networks (taxi and metro network)
that we use in our experimental evaluation.
MAD is defined as $\frac{\Sigma_{x\in X}\mid{x-\mu}\mid}{|X|}$, where
$X$ is the set of samples and $\mu$ is their averages.
The computed MAD for taxi and metro network is $0.10$ and $0.06$ respectively.
Hence, we can map each trip in the trips table
to an {\em atomic ODT triple} $(o,d,t)$, where
$o$ is the origin region of the trip (if the origin is a GPS location, it can
be mapped to the nearest $v\in V$), $d$ is the destination region of
the trip, and $t$ is the atomic timeslot that contains the
origin time of the trip.%
\footnote{Our definitions and techniques can be extended for data inputs
where there is no correlations between the begin and end times of
trips; in this case, we should map trips to $(o,d,st,et)$ quadruples
and patterns should be ODSE quadruples.}

\begin{definition} [Atomic ODT triple] \label{def:atriple}
  A triple $(o,d,t)$ is atomic if:
  \begin{itemize}
  \item $o$ is an atomic region
  \item $d$ is an atomic region
  \item $t$ is an atomic timeslot
  \item $o \neq d$
  \end{itemize}
\end{definition}

Given an atomic ODT triple $P$, the
{\em support} $\sigma(P)$ of $P$ is the total number of passengers
(flow) of the trips that are mapped to $P$.
For example, the top-left of Figure \ref{fig:detailed} shows a map and
an individual trip in it, which corresponds to the first trip in
Figure \ref{fig:table}. The top-right of Figure \ref{fig:detailed} has
the {\em aggregated trips table}, which contains all atomic $(o,d,t)$
triples, after aggregating all trips that correspond to the same
$(o,d,t)$. For instance, trips ($B$, $D$, 9:20, 2) and ($B$, $D$,
9:29, 1) are merged to triple $(B,D,18)$ with total flow $3$.%
\footnote{All time moments between
  9:00 and 9:30 are generalized to timeslot 18, which is the 18th slot
  in 30-minute intervals, starting from 00:00-00:30 (mapped to 0).}
Next, we define generalized (i.e., non-atomic) ODT
triples.




\begin{definition} [ODT triple] \label{def:etriple}
An ODT triple $(O,D,T)$ consists of a region $O$, a region $D$, and a
timeslot $T$, such that $O\cap D = \emptyset$.
\end{definition}

\begin{definition} [ODT triple generalization] \label{def:gpattern}
  An ODT triple $P_1$ is a generalization of ODT triple $P_2$ if
for each $X\in \{O,D,T\}$, either $P_1.X=P_2.X$ or $P_1.X\subset
P_2.X$,
and for at least one $X\in \{O,D,T\}$, $P_2.X\subset
P_1.X$.
\end{definition}

\begin{definition} [Minimal generalization of ODT triple] \label{def:mingen}
  An ODT triple
  $P_1$ is a minimal generalization of ODT triple $P_2$ if one of the
  following holds:
    \begin{itemize}
  \item $P_1.O=P_2.O$, $P_1.D=P_2.D$ and $P_1.T$ is a minimal generalization of $P_2.T$  \item $P_1.D=P_2.D$, $P_1.T=P_2.T$ and $P_1.O$ is a minimal generalization of $P_2.O$  \item $P_1.O=P_2.O$, $P_1.T=P_2.T$ and $P_1.D$ is a minimal generalization of $P_2.D$ \end{itemize}
\end{definition}


\begin{figure}[h]
      \centering
      \includegraphics[width=1.00\columnwidth]{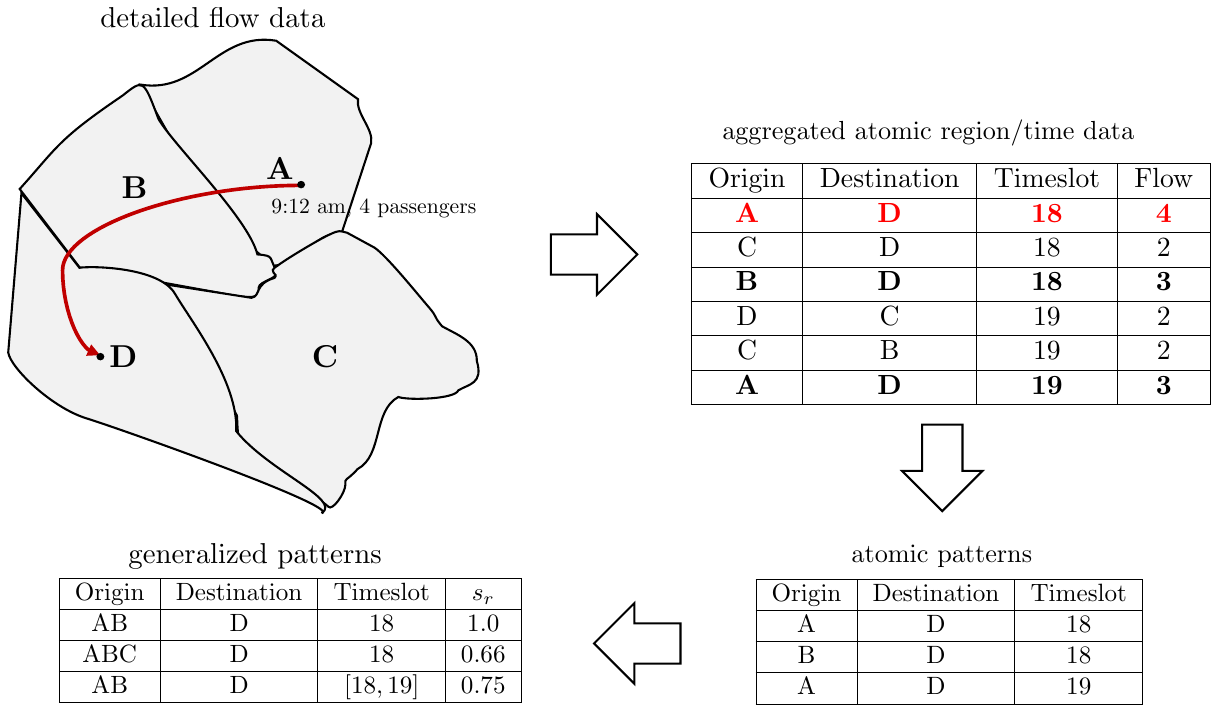}
      \caption{A detailed example}
      \label{fig:detailed}
\end{figure}


We now turn to the definition of ODT patterns; we start by atomic ODT
patterns and then move to the generalized  ODT patterns.

\begin{definition} [Atomic ODT pattern] \label{def:atomicpattern}
Let $AT_r$ be the set of atomic ODT triples with non-zero support.
Given a threshold $s_a, 0< s_a\le 1$, an
atomic ODT triple $P$ is called an atomic ODT pattern if $\sigma(P)$ is
in the top $s_a\times |AT_r|$ supports of triples in $AT_r$.
\end{definition}

Figure \ref{fig:detailed} (bottom-right) shows the atomic ODT patterns for our running example if $s_a=0.5$.
The above definition considers a global support threshold for
characterizing an atomic triple as a pattern,  following the typical
approach in data mining.

\begin{definition} [ODT pattern] \label{def:pattern}
An ODT pattern $P$ is an ODT triple where:
\begin{itemize}
\item the ratio of atomic triples in $P$, which are atomic patterns, is at
least equal to a minimum ratio threshold $s_r$
\item there exists a minimal specialization of $P$ which is an ODT pattern  
\end{itemize}  
\end{definition}

The number of atomic triples in $P$, which are atomic patterns is
denoted by  $P$.cnt.
In the example of Figure \ref{fig:detailed}, 
if $s_r=0.6$, 
$(AB,D,18)$ is a (generalized) ODT pattern where origins $A$ and $B$
have significant joint flow to destination $D$ at timeslot $t=18$,
because the pattern includes two out of three atomic patterns. The
second condition of Def. \ref{def:pattern} is a {\em sanity} constraint,
which prevents a generalized triple $P$ from being characterized as a pattern
if there is no minimal specialization of $P$ that is also a pattern;
intuitively, a pattern should have at least one minimal specialization which
is also a pattern (weak monotonicity).

A pattern (triple) $P$ is said to be level-$\ell$ pattern (triple) if the total number of
atomic elements in it (regions and timeslots) is $\ell$. Hence, atomic
patterns are level-3 patterns, since they contain exactly 3 elements
(i.e., two atomic regions and one atomic timeslot). Similarly, triple
$(A, BC, [1,3])$ is a level-6 triple because it includes 1 atomic
region in its origin, 2 atomic regions in its destination, and 3
atomic regions in its time-range (note that $[1,3]$ includes atomic
timeslots $\{1,2,3\}$).


%% file: algorithms.tex
\section{Pattern Extraction}\label{sec:algo}
To find the ODT patterns, 
we first start by finding the atomic ODT patterns, i.e., the $(o,d,t)$
triples which are frequent/significant, where $o$ and $d$ are atomic
regions and $t$ is an atomic timeslot.
This is trivial and can be done by one pass over the aggregated trips data, where the
occurrence of each $(o,d,t)$ triple is unique and by selecting the top
$s_a$ ratio of them as atomic patterns. 
Our most challenging task is then to define an algorithm that
progressively synthesizes non-atomic patterns from atomic patterns and
prunes the search space effectively.

Recall that a non-atomic triple $P=(O,D,T)$ is a pattern
if at least a ratio $s_r>0$ of its included atomic triples
are patterns. Hence, by definition, a non-atomic pattern generalizes at
least one atomic pattern $(o,d,t)$.
The
pattern synthesis algorithm uses the set of atomic patterns
and the
region neighborhood graph $G$ to synthesize the non-atomic patterns.
Given an existing $(O,D,T)$ pattern $P$ of size $k$, we
attempt a minimal generalization of $P$ by including into the
set $O$ a neighboring atomic region to the existing regions in $O$, or
doing the same for set $D$, or adding an atomic neighboring timeslot
to $T$.

The challenge is to
prune candidate generalizations that cannot be patterns.
For this, we need a fast way to compute (or bound) the number of
contributing (newly added to $P$) atomic patterns to the ratio of the
candidate.



\subsection{Baseline Algorithm}\label{sec:baseline}
We now present a baseline algorithm for enumerating all the atomic and
extended ODT patterns in an input graph $G(V,E)$.
The first step of Algorithm \ref{algo:baseline} is to scan all trips
data and compute the support counts of all atomic triples $\mathcal{T}_3$.
Then, it finds the set $\mathcal{P}_3$ of atomic patterns, i.e., the triples
having support count at least equal to
$minsup$, which is the support
count of the $s_a\cdot|\mathcal{T}_3|$-th triple in $\mathcal{T}_3$ with the highest
support.
All triples (patterns) in $\mathcal{T}_3$ ($\mathcal{P}_3$) have exactly
three atomic elements (regions or timeslots).
The algorithm progressively finds the patterns with more atomic
elements.
Recall that a triple (pattern) $P$ is at level $\ell$, i.e., in set $\mathcal{T}_\ell$
($\mathcal{P}_\ell$) if it has $\ell$ atomic elements; we also call $P$ an
$\ell$-size triple (pattern).
Candidate patterns $CandP$ at level $\ell+1$ are generated by either adding an atomic
region at $O$ or an atomic region at $D$ or an atomic timeslot at $T$,
provided that the resulting triple is valid according to Definition \ref{def:etriple}.
If a $CandP$ has been considered before, it is disregarded. This may
happen because the same triple can be generated from two or more
different triples at level $\ell$. For example, candidate pattern
$(AB,C,1)$ could be generated by pattern $(A,C,1)$ (by extending
region $A$ to region $AB$) and by $(B,C,1)$ (by extending
region $B$ to region $AB$).
Hence, we keep track at each level $\ell$ the set of triples that
have been considered before, in order to avoid counting the same
candidate twice.%
\footnote{Since a pattern at level $\ell+1$ requires at least one
  and not all its minimal specializations to be patterns, an id-numbering
  scheme for atomic regions, which would extend patterns by only
  adding elements that have larger id would not work. For example, if
  both $(A,C,1)$ and  $(B,C,1)$ are patterns, $(AB,C,1)$ can be
  generated by both of them; however, if just $(B,C,1)$ is a pattern,
  $(AB,C,1)$ can only be generated by  $(B,C,1)$.}

To check whether a candidate $CandP$ not considered before is a
pattern,
we need to divide the number $CandP$.cnt of atomic patterns included in $CandP$
by the total  number $CandP$.card of atomic triples in $CandP$.
If this ratio is at least $s_r$, then $CandP$ is a pattern.
$CandP$.card can be computed algebraically: it is the product of
atomic elements in each of the three ODT components. For example,
$(AB,CD,[1,3])$.card = $2\cdot 2\cdot 3=12$ because there are 12
atomic triples in $(AB,CD,[1,3])$, i.e., combinations of elements
$\{A,B\}$,  $\{C,D\}$, and $\{1,2,3\}$.
To compute $CandP$.cnt fast, we can take advantage of the fact that we
already have $P$.cnt, i.e., the number of atomic patterns in the
generator pattern. We only have to compute the $P'$.cnt for the
{\em difference} $P'=CandP-P$ between $CandP$ and $P$, which is the
triple consisting of the extension element in the extended dimension
(one of O, D, T) together with the element-sets in the intact
dimensions (two of O, D, T).
For example, if $P=(A,CD,[1,2])$ and $CandP=(AB,CD,[1,2])$, then
$P'=(B,CD,[1,2])$.
To compute $P'$.cnt,
Algorithm \ref{algo:baseline} enumerates all atomic triples in $P'$ to
check whether they are patterns.
It then sums up $P$.cnt and $P'$.cnt to derive $CandP$.cnt.


\begin{algorithm}
\begin{algorithmic} [1]
  \small
\Require a region graph $G(V,E)$; a trips table;
a minimum support $s_a$ for atomic ODT patterns; a minimum support
ratio $s_r$ for
non-atomic ODT patterns 
\State $\mathcal{T}_3$ = atomic triples computed from trips table
\State $\mathcal{P}_3$ = triples in $\mathcal{T}_3$ with support $\geq s_a$
\For{all atomic triples $P\in \mathcal{T}_3$}
       \State $P$.cnt = 1 if $P\in \mathcal{P}_3$, else $P$.cnt=0
 \EndFor   
\State $\ell$ = 3
\While {$|\mathcal{P_\ell}| > 0$}  \Comment{$\mathcal{P}_\ell$ = set of
  level-$\ell$ patterns}          
     \State $\mathcal{P}_{\ell+1}$ = $\emptyset$   \Comment{Initialize pattern set at level
       $\ell+1$}
     \For{each $P$ in $\mathcal{P}_\ell$}
         \For{each minimal generalization $CandP$ of $P$}
               \If{$CandP$ not considered before}
                      \State $P'$= $CandP-P$ 
                      \State $CandP$.cnt = $P$.cnt + $P'$.cnt
                     \If{$CandP$.cnt $/$ $CandP$.card $\geq s_r$}
                           \State add $CandP$ to $\mathcal{P}_{\ell+1}$
                       \EndIf
                 \EndIf
            \EndFor
     \EndFor
     \State $\ell$ = $\ell$ + 1       
\EndWhile                                 
\end{algorithmic}
\caption{Baseline Algorithm for finding all ODT patterns}
\label{algo:baseline}
\end{algorithm}

Figure \ref{fig:lattice} exemplifies the pattern enumeration
process in our running example (see Figure \ref{fig:detailed}).
Atomic pattern $P=(A,D,18)$ can be generalized by adding to the origin
any of the neighbors of atomic region $A$, to the destination any of the neighbors
of atomic region $D$, and to timeslot 18 either timeslot 17 or
timeslot 19. Each of these generalization forms a candidate pattern
$CandP$ at level 4. Counting the support of these candidates requires
counting only the difference $P'$. For example, to count the support
of $(AB,D,18)$, we only have to add to the support of $P=(A,D,18)$ the
support of $P'=(B,D,18)$, which is 1. Then, the support of $(AB,D,18)$
is found to be 2. Assuming that $s_r=0.6$, $CandP=(AB,D,18)$ is a pattern,
since the ratio of atomic patterns in it is $1.0\ge s_r$. All patterns
that stem from $P=(A,D,18)$ up to level 5 are emphasized in Figure
\ref{fig:lattice}; these are used to generate candidate patterns at
the next levels. 

\begin{figure}[h]
      \centering
      \includegraphics[width=1.04\columnwidth]{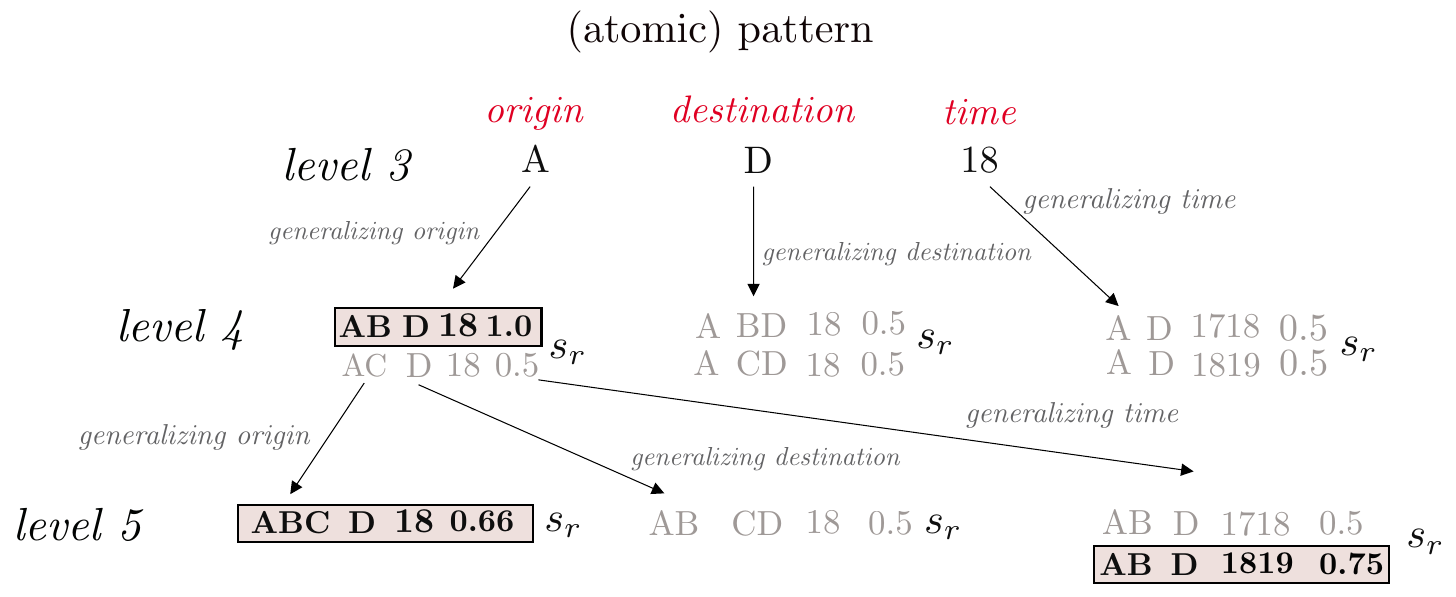}
      \caption{Pattern enumeration example}
      \label{fig:lattice}
\end{figure}

\stitle{Complexity Analysis}
In the worst case, all valid combinations of regions and timeslots can
be considered as candidate patterns.
In other words, each subset of $V$ can be the union of O and D and for
each such subset of size $k$ can be split in $2^k-2$ ways.
Hence, the number of possible OD pairs is 
$\sum_{k=1}^{|V|-1}{|V| \choose k}(2^k-2)$.
If $S$ is the number of atomic timeslots, then the number of
possible (generalized) timeslots to be included in a candidate pattern
is $2^{|S|}$. Hence, the worst-case space/time complexity of ODT pattern
enumeration is $O\left
  (2^{|S|}\sum_{k=1}^{|V|-1}{|V| \choose k}(2^k-2)\right )$.
The complexity increases exponentially with the number of
atomic regions and the number of atomic timeslots, rendering the
problem particularly hard. In following, we propose methods
that reduce the complexity in practice.

\subsection{Optimizations}\label{sec:opt}
We now discuss some optimizations to the baseline algorithm, which can
greatly enhance its performance.

{\bf Avoid re-counting $P'$.}
The first approach is based on {\em caching} the ODT triples that have
been counted before. Instead of computing $P'$.cnt directly for
$P'=CandP-P$, we first check whether $P'$.cnt is already
available. This requires us to cache the counted triples and their
supports at each level in a hash table. Hence, before counting $P'$,
we first search the hash table, which caches the triples of size
$|P'|$ to see if $P'$ is in there. In this case, we simply use
$P'$.cnt instead of computing it again from scratch.

{\bf Fast check for zero support of $P'$.}
The second optimization is based on the observation that for some
pairs $(o,d)$ of atomic regions,
there does not exist any timeslot $t$, such that
$(o,d,t)$ is an atomic pattern in $\mathcal{P}_3$. For example, if $o$ and $d$
are remote regions on the map, it is unlikely that there is
significant passenger flow that connects them at any time of the day.
We take advantage of this to avoid counting any $P'$ which may not
include atomic patterns. Specifically, for each atomic region $r$, we
record (i) $r.dests$, the set of atomic regions $r'$, such that there
exists a $(r,r',t)$ pattern in $\mathcal{P}_3$; and (ii)
$r.srcs$, the set of atomic regions $r'$, such that there
exists a $(r',r,t)$ pattern in $P_3$.
If, in Algorithm $\ref{algo:baseline}$, $CandP$ is a minimal
generalization of $P$, by expanding $P.O$ to include a new atomic
region $r$, then $P'=CandP-P$ should only include $r$ in $P'.O$. If
$P'.D\cap r.dests = \emptyset$, then $P'$ does not include any atomic
patterns and $CandP$.cnt is guaranteed to be equal to $P$.cnt. Hence,
we can skip support computations for $P'$.
Symmetrically, if $CandP$ is a minimal
generalization of $P$, by expanding $P.D$ to include a new atomic
region $r$, then $P'=CandP-P$ should only include $r$ in $P'.D$. If
$P'.O\cap r.srcs = \emptyset$, then $CandP$.cnt is guaranteed to
be equal to $P$.cnt.
Overall, by keeping track of $r.dests$ and $r.srcs$ for each atomic
region $r$, we can save computations when counting the supports of patterns. 
Since the space required to store $r.dests$ and $r.srcs$ is $O(|V|)$
in the worst case, the total space complexity of these sets is
$O(|V|^2)$. This cost is bearable, because our problem typically
applies on transportation networks or district neighborhood graphs in
urban maps,  where the number of vertices in $V$ is rarely large.  


{\bf Improved neighborhood computation.}
The minimal generalizations of a pattern $P$ are generated by
minimally generalizing $P.O$, $P.D$, and $P.T$.
The generalization of $P.T$ is trivial as we add one atomic
timeslot before the smallest one in  $P.T$ or after the largest one in
$P.T$.
On the other hand, computing the minimal generalizations of a region
$R$ (i.e., $P.O$ or $P.D$) can be costly if done in a brute-force
way. The naive algorithm tries to add to $P.O$ all possible neighbors of each
atomic region $r\in R$ and for each such neighbor not in $P.O$ and
$P.D$ it measures the support of the corresponding generalized
pattern $P'$, if $P'$ was not considered before. Since the same $P'$
can be generated by multiple $P$,  checking whether  $P'$ has been
considered before can be performed a very large number of times with a
negative effect in the runtime. We design a neighborhood computation
technique for a region $R$, which avoids generating the same $P'$
multiple times. The main idea is to collect first all neighbors of all
$r\in R$ in a set $N$ and then compute (in one step) $N-P.O-P.D$,
i.e., the set of regions $r$ that minimally expand $R$ to form the
minimal generalizations of $P$. 

{\bf Indexing atomic patterns.}
As another optimization, we employ a prefix-sum index which can help
us to compute an upper bound of the support of $P'$.
The main idea comes from indexes used to compute range-sums in OLAP
\cite{DBLP:conf/sigmod/HoAMS97}.
Let $N$ be the number of atomic regions and $M$ be the number of
atomic timeslots.
Consider a $N\times N \times M$
array $A$, where each cell corresponds to an atomic ODT
triple.
The cell includes a 1 if the corresponding   atomic ODT
triple is a pattern; otherwise the cell includes a  0.
In addition, consider a 3D array $R$ with shape $(N+1)\times (N+1) \times (M+1)$.
Each element $R[i][j][k]$ of $R$ is the sum of all elements
$A[i'][j'][k']$ of $A$, such that $i'\le i$, $j'\le j$, and $k'\le k$
$R[i][j][k]=0$ if any of $i,j,k$ is 0. $R$ is the {\em
  prefix-sum} array of $A$.
Figure \ref{fig:prefix} illustrates the prefix sum 3D array $R$.

\begin{figure}[h]
      \centering
      \includegraphics[width=1.00\columnwidth]{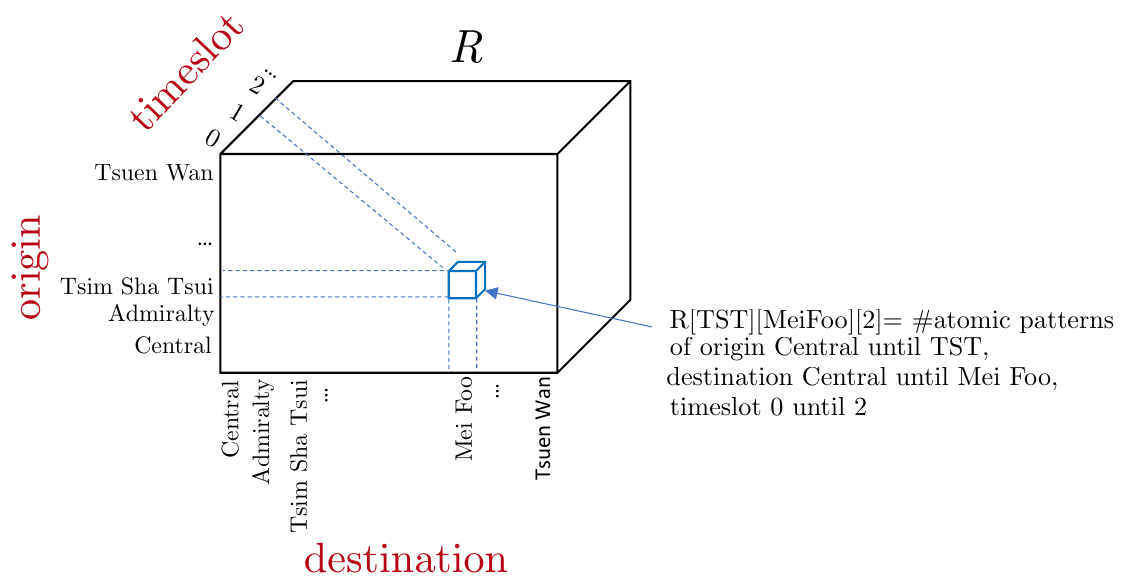}
      \caption{Prefix sum example}
      \label{fig:prefix}
\end{figure}

Now consider a 3D range $[a,b], [c,d], [e,f]$, where $0<a\le b \le N$,
$0<c\le d \le N$, and $0<e\le f \le M$
and assume that the objective is to compute the sum of values in $A$
inside this range.
We can show that this sum can be accumulated by seven computations as
follows:
\begin{align*}
&R[b][d][f]\\
&-R[a-1][d][f]-R[b][c-1][f]-R[b][d][e-1]\\
&+R[a-1][c-1][f]+R[b][c-1][e-1]+R[a-1][d][e-1]\\
&-R[a-1][c-1][e-1]
\end{align*}


Now, consider a $P'$ which needs to be counted. $P'$ includes a set of
atomic origin regions, a set of atomic destination regions, and a set
of atomic timeslots.
The atomic timeslots are guaranteed to be a continuous sublist of
regions in the corresponding dimension of the 3D array $A$, starting,
say, from timeslot $e$ and including up to timeslot $f$.
However, the region sets in $P'$ are not guaranteed to be continuous.
Still, the 3D range $[a,b], [c,d], [e,f]$, where $a$ ($c$) is the
origin (destination) region in $P'$ with the smallest ID and 
$b$ ($d$) is the
origin (destination) region in $P'$ with the largest ID is
guaranteed to be a superset of atomic
triples in $P'$.
Hence, the prefix-sum index can give us in $O(1)$ time an {\em upper
  bound} of the number of atomic patterns in $P'$.
If this upper bound is added to the support of $P$ and the resulting
support is less than $s_r$, then $CandP$ is definitely not a pattern,
so we can avoid counting $P'$.
To maximize the effectiveness of this optimization, we should find a
total order of the regions which preserves as much as possible the
continuity of space, that is, try to give IDs to regions which
follows the region generalization.

%% file: extensions.tex
\section{Pattern Variants}\label{sec:ext}
In this section, we explore alternative problem definitions and the
corresponding problem solutions that can be more useful than our
general definition in certain problem instances.
In particular, we observe that the number of patterns can be huge even
if relatively small $s_a$ and large
$s_r$ are used.
In addition, setting global
thresholds may not be ``fair'' for some regions which are
under-represented in the data. 
To address these issues, we propose (i) size-bounded patterns, (ii)
constrained-pattern search,
and (iii) rank-based patterns.

\subsection{Size-bounded Patterns}\label{sec:sizebounded}
The first type of constraint that we can put to limit the number of
patterns is on the size of the regions or timeslots in a
pattern. Specifically, we can set an upper bound $B_O$ to $|O|$, i.e., the
number of atomic regions in an origin region of a pattern. Similarly,
we can limit the number of regions in $D$ to at most $B_D$ and the
number of atomic timeslots to at most $B_T$. In effect, this limits
the number of levels that we use for pattern search to $B_O\cdot
B_D\cdot B_T$ and reduces the number of patterns at each level. 

For pattern enumeration, we use the same algorithms and optimizations
discussed in Section \ref{sec:algo}, but with the constraints applied
whenever we expand a pattern to generate the candidate patterns at the
next level.

\subsection{Constrained Patterns}\label{sec:restricted}

Another way to control the number of the patterns, but also focus on
specific regions and/or timeslots that are under-represented in the
entire population is to limit the domain of atomic regions and
timeslots.
Specifically, we give as parameter to the problem the set of
atomic regions $V_O\subseteq V$ that we are interested in to serve as
origins the set $V_D\subseteq V$ of regions that can serve as
destinations, and $T_R\subseteq T$, a restricted contiguous
subsequence of the entire sequence of atomic timeslots $T$ to be used as
timeslots in the patterns.
The induced subgraphs by $V_O$ and $V_D$ should be connected, in order
to potentially  have the entire $V_O$ (and/or  $V_D$) as an origin
(destination) of a pattern.
For example, if a data analyst is interested in flow patterns from South
Manhattan to Queens in afternoon hours, she could include in $V_O$
(resp. $V_D$) all the atomic regions in South
Manhattan (resp. Queens) and restrict the timeslots to be used in
patterns to only afternoon hours.

Recall that thresholds $s_a$ and $s_r$ apply to the set of atomic
triples and ratio of atomic patterns, respectively. Hence, by
constraining $V_O$, $V_D$, and $T_R$, we consider only atomic triples
for the regions (times) of interest, making it possible to detect
patterns that are under-represented in the entire set of atomic
triples. For example, restricting $V_O$ to be a remote district on
the map, makes it possible to detect flow patterns from that district,
which would not be found otherwise, assuming that the outgoing flow
from that district is very small compared to the outgoing flow from
all other districts.

Again, adapting our pattern enumeration algorithm and its variants to
identify constrained patterns is straightforward, as we only have to
(i) confine atomic triples and patterns to include only origins in
$V_O$, destinations in $V_D$, and timeslots in $T_R$, and (ii) limit
the expansion of regions/timeslots in candidate pattern generation,
according to the constraints  $V_O$, $V_D$, and $T_R$.
Depending on the sizes of $V_O$, $V_D$, and $T_R$ pattern enumeration
can be significantly faster compared to unconstrained pattern search.





\subsection{Rank-based patterns}\label{sec:rank}
Another way to control the number of patterns and still not miss the most
important ones is to regard as patterns, at each level,
the $k$ triples with the highest support 
and prune the rest of them as non-patterns.
This is achieved by replacing the minimum ratio threshold $s_r$ by
a parameter $k$, which models the ratio of eligible triples
at each level which are considered to be important.

More formally, let $\mathcal{T}_\ell$ be the set of triples at level
$\ell$, which are minimal generalization of patterns at level $\ell-1$.
The set of patterns $\mathcal{P}_\ell$ at level $\ell$ consists of the 
$k$ triples in $T_{\ell}$ having the largest number of atomic patterns.

\begin{definition} [ODT pattern (rank-based)] \label{def:pattern}
An ODT triple $P$ at level $\ell$ is a rank-based ODT pattern if:
\begin{itemize}
\item there exists a minimal specialization of $P$ which is a
  rank-based ODT pattern  
\item there are no more than $k$ minimal generalizations
  of level-$(\ell-1)$ rank-based ODT patterns that include more frequent
atomic patterns than $P$. 
\end{itemize}  
\end{definition}





\subsubsection{Baseline approach for rank-based pattern enumeration}
\label{sec:baselinerank}
A baseline approach for enumerating rank-based patterns is to
generate all eligible triples at each level $\ell$, which are minimal
generalizations of patterns at level $\ell-1$. For each such triple,
count its support (i.e., number of atomic patterns included in
it). We may use the optimizations proposed in Section \ref{sec:opt},
to reduce the cost of generating ODT triples that are candidate
patterns and counting their supports.
After generating all triples and counting their
supports, we select the top-$k$ ones as patterns.
Only these patterns are used to generate the candidate patterns at
level $\ell +1$.

\subsubsection{Optimized rank-based pattern enumeration}
\label{sec:optrank}
To minimize the number of generated triples at each level $\ell$ and
the effort for counting them, we
examine the patterns at $\ell-1$ in decreasing order of their
potential to generate triples that will end up in the top-$k$
triples at level $\ell$.
Hence, we access the patterns $P$ at level $\ell-1$  in decreasing order
of their support $P$.cnt. For each such pattern $P$ and for each minimal
generalization $CandP$ of $P$, we first compute the potential of
$P'=CandP-P$ to add to the support $CandP$.cnt (initially $CandP$.cnt
= $P$.cnt). If, by adding the maximum possible $P'$.cnt to
$CandP$.cnt, $CandP$.cnt cannot make it to the top-$k$ $\ell$-triples
so far, then we prune $CandP$ and avoid its counting.
The maximum possible $P'$.cnt can be computed based on the following
lemma:

\begin{lemma}\label{lemma:maxsupport}
The maximum possible $P'$.cnt that can be added to $P$.cnt, to derive
the support of $CandP$ is as follows:
\begin{itemize}
  \item If $CandP$ is generated by minimally generalizing $P.O$, then
    $P'$.cnt equals $|P.D|\cdot |P.T|$.    
  \item If $CandP$ is generated by minimally generalizing $P.D$, then
    $P'$.cnt equals $|P.O|\cdot |P.O|$.
  \item If $CandP$ is generated by minimally generalizing $P.T$, then
    $P'$.cnt equals $|P.O|\cdot |P.D|$.   
\end{itemize} 
\end{lemma}
\begin{proof}
 Each of the three cases is proved as follows:
\begin{itemize}
  \item If $CandP$ is generated by minimally generalizing $P.O$, then
    $P'.O$ is an atomic region, $P'.D=P.D$, and $P'.T=P.T$; hence, the
    maximum possible $P'$.cnt equals $|P.D|\cdot |P.T|$.    
  \item If $CandP$ is generated by minimally generalizing $P.D$, then
    $P'.D$ is an atomic region, $P'.O=P.O$, and $P'.T=P.T$; hence, the
    maximum possible $P'$.cnt equals $|P.O|\cdot |P.O|$.
  \item If $CandP$ is generated by minimally generalizing $P.T$, then
    $P'.T$ is an atomic timeslot, $P'.O=P.O$, and $P'.D=P.D$; hence, the
    maximum possible $P'$.cnt equals $|P.O|\cdot |P.D|$. 
\end{itemize} 
\end{proof}

Let $\theta$ be the $k$-th largest support of the 
triples generated so far at level $\ell$.
If for the next examined $P$ from level $\ell-1$ to generalize,
$P$ cannot be generalized to a $CandP$ that may end up in the top-$k$
level-$\ell$ triples, then we can immediately prune $P$. The condition
for pruning $P$ follows:

\begin{lemma}\label{lem:prunep}
If $P$.cnt + $\max\{|P.D|\cdot |P.T|, |P.O|\cdot |P.O|, |P.O|\cdot
|P.D|\}$ $\le\theta$, then no minimal generalization of $P$ can enter
the set of top-$k$ level-$\ell$ ODT triples.
\end{lemma}
\begin{proof}
  The proof stems directly from Lemma \ref{lemma:maxsupport}.
  Any candidate pattern $CandP$ which is a minimal generalization of $P$
  belongs to one of the three cases above. Hence, the maximum possible
  support for $CandP$  is $P$.cnt plus the maximum of
  the three products that $P'$.cnt can be.
\end{proof}

\begin{algorithm}
\begin{algorithmic} [1]
 \scriptsize
\Require a region graph $G(V,E)$; a trips table;
a minimum support $s_a$ for atomic ODT patterns; number $k$ of top
patterns to be generated at each level; maximum level considered ($maxl$) 
\State $\mathcal{T}_3$ = atomic triples computed from trips table
\State $\mathcal{P}_3$ = triples in $\mathcal{T}_3$ with support $\geq s_a$
\For{all atomic triples $P\in \mathcal{T}_3$}
       \State $P$.cnt = 1 if $P\in \mathcal{P}_3$, else $P$.cnt=0
 \EndFor   
\State $\ell$ = 3
\While {$|\mathcal{P_\ell}| > 0$ and $\ell< maxl$}  \Comment{extend
  level-$\ell$ patterns}          
     \State $\mathcal{P}_{\ell+1}$ = $\emptyset$   \Comment{Initialize
       $k$-minheap with level-$(\ell+1)$ patterns}
     \For{each $P$ in $\mathcal{P}_\ell$ in decreasing order of
       $P$.cnt}
         \If {$|\mathcal{P}_{\ell+1}|=k$ and $P$.cnt+$\max\{|P.D|\cdot |P.T|, |P.O|\cdot |P.O|, |P.O|\cdot
           |P.D|\}\le \mathcal{P}_{\ell+1}$.top.cnt}
         \State {\bf continue} \Comment{Prune $P$ based on Lemma \ref{lem:prunep}}
         \EndIf
         \If {$|\mathcal{P}_{\ell+1}|=k$ and $P$.cnt+ $|P.D|\cdot
           |P.T|\le \mathcal{P}_{\ell+1}$.top.cnt}
            \For{each minimal generalization $CandP$ of $P$ by origin}
               \If{$CandP$ not considered before} \label{lin:st}
                  \State $P'$= $CandP-P$
                  \State$CandP$.cnt = $P$.cnt + $P'$.cnt
                  \If{$|\mathcal{P}_{\ell+1}|<k$}
                     \State add $CandP$ to $\mathcal{P}_{\ell+1}$
                  \Else
                     \If{$CandP$.cnt $>\mathcal{P}_{\ell+1}$.top.cnt}
                     \State update $\mathcal{P}_{\ell+1}$ with $CandP$
                     \EndIf
                  \EndIf
               \EndIf \label{lin:end}
            \EndFor
          \EndIf
          \If {$|\mathcal{P}_{\ell+1}|=k$ and $P$.cnt+ $|P.O|\cdot
           |P.T|\le \mathcal{P}_{\ell+1}$.top.cnt}
            \For{each minimal generalization $CandP$ of $P$ by
              dest.}
            \State Lines \ref{lin:st} to \ref{lin:end} above 
            \EndFor
            \EndIf
          \If {$|\mathcal{P}_{\ell+1}|=k$ and $P$.cnt+ $|P.O|\cdot
           |P.D|\le \mathcal{P}_{\ell+1}$.top.cnt}
            \For{each minimal generalization $CandP$ of $P$ by
              time}
            \State Lines \ref{lin:st} to \ref{lin:end} above 
            \EndFor
         \EndIf  
     \EndFor
     \State $\ell$ = $\ell$ + 1       
\EndWhile                                 
\end{algorithmic}
\caption{Optimized Algorithm for enumerating rank-based ODT patterns}
\label{algo:rankbased}
\end{algorithm}

Based on the above lemmas, we can prove the correctness of our
enumeration algorithm for rank-based ODT patterns, described by
Algorithm \ref{algo:rankbased}.
The algorithm computes first all level-$3$ patterns $\mathcal{P}_3$, based on the
atomic pattern support threshold $s_a$ (Lines 3--5).
Having the patterns at level $\ell$, the algorithm organizes those at
level $\ell + 1$ in a priority queue (minheap)  $\mathcal{P}_{\ell+1}$ of
maximum size $k$.
We consider all patterns $P$ at level $\ell$ in decreasing order of
support  $P$.cnt, to maximize the potential of generating level-$(\ell
+ 1)$ triples of high support early.
For each such pattern $P$, we first check if $P$ can
generate any level-$(\ell+1)$ triple that can enter the set $\mathcal{P}_{\ell+1}$ of
top-$k$ triples so far at level $\ell+ 1$, based on Lemma
\ref{lem:prunep}.
If this is not possible, then $P$ is pruned.
Otherwise, we attempt to generalize $P$, first by adding an atomic
region to $P.O$. If the maximum addition to $P$.cnt by such an
extension cannot result in a $CandP$ that can enter the top-$k$ at
level $\ell+ 1$ (based on Lemma \ref{lemma:maxsupport}), then we do
not attempt such extensions; otherwise we try all such extensions and
measure their supports (Lines \ref{lin:st} to \ref{lin:end}).
We repeat the same for the possible extensions of $P.D$ and $P.T$.
After $\mathcal{P}_{\ell+1}$ has been finalized, we use it to generate
the top-$k$ patterns at the next level.
Since the number of levels for which we can generate patterns can be
very large, Algorithm \ref{algo:rankbased} takes as a parameter the
maximum level $maxl$ for which we are interested in generating patterns.

%% file: experiments.tex
\section{Experiments}\label{sec:exps}
In this section, we evaluate the performance of our proposed
algorithms on real datasets.
All methods were implemented in Python3 and the experiments were run on
a Macbook Air with a M2 processor and 16GB memory. The source code of
the paper is publicly
available\footnote{https://github.com/kosyfakichry/spatiotemporalflowpatterns}.

\subsection{Dataset Description}
For our experiments, we used three real datasets; NYC
taxi trips, a metro network trips
and Flights.
Below,
we provide a detailed description for each of them.

\stitle{NYC taxi trips:} We processed 7.5M trips of yellow taxis in NYC
in January 2019, downloaded from TLC%
\footnote{https://www.nyc.gov/site/tlc/about/tlc-trip-record-data.page}.
Each record represents a taxi trip and
includes the pick-up and drop-off taxi zones
(different regions in NYC), the date/time of the pick-up, and the
number of passengers who took the trip.
We converted all time moments to 48 time-of-day slots (one slot per
30min intervals in the 24h). Then, we aggregated the data by merging
all trips having  the same origin, destination, and
timeslot, and summing up the total number of passengers in all these
trips to a total passenger flow, as explained in Section
\ref{sec:def}. This way, we ended up having 373460 unique ODT combinations
 (atomic ODT triples), which we used as input to our pattern
 enumeration algorithms.
In addition, we used the maps posted at the same website to construct
the neighboring graph $G$ between the atomic regions (taxi zones).
In $G$, we connected all pairs of atomic regions that share boundary
points or are separated by water boundaries.

		


\stitle{Metro trips:} We extract data from a metro network.
The system consists of 168
stations, serving a number of areas. We consider each station as an
atomic region; we created the neighborhood graph $G$ for them by
linking stations that are next to each other in the network.  The data
are aggregated for all passenger trips taken in September
2019. Specifically, for each atomic ODT triple, where the origin and
destination are stations and T is one of the 48 atomic timeslots, we
have the total number of passenger trips in Sep. 2019. The total
number of atomic  ODT triples is 253497.


\stitle{Flights:} We extracted information for 5.8M US flights in 2015 from Kaggle%
\footnote{https://www.kaggle.com/datasets/usdot/flight-delays?select=flights.csv}.
In
this dataset, we consider as atomic regions 319 airports in North
America that appear in the file.
Since the number of passengers in each flight was not given in the original data, we randomly generated a number between 50 and 200.
We followed the same procedure as in for the two previous
datasets; namely, we converted the original flights data into a table with
atomic ODT triples.
The total number of resulting ODT triples is 17623.
To
create the neighbor graph $G$, we
follow the same logic as the two previous datasets; we connect
atomic regions in neighboring states. 

\begin{figure*}[t!]
  \centering
\vspace{-2mm}
 \subfigure[Taxi Network]{
   \label{fig:exp:sataxi}
    \includegraphics[width=0.3\textwidth]{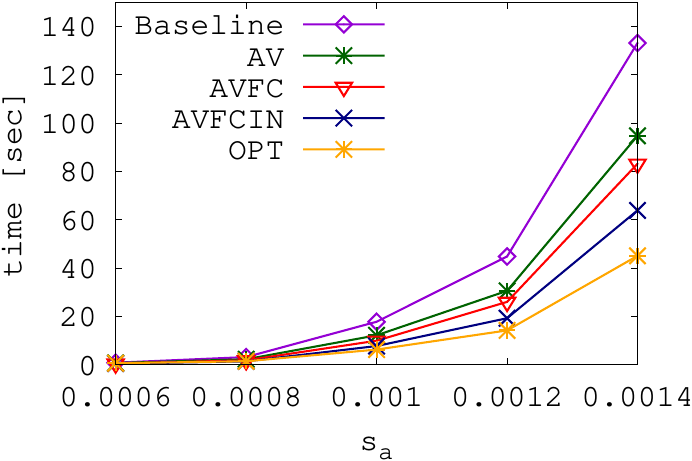}
    }\!\!
  \subfigure[Metro Network]{
   \label{fig:exp:samtr}
    \includegraphics[width=0.3\textwidth]{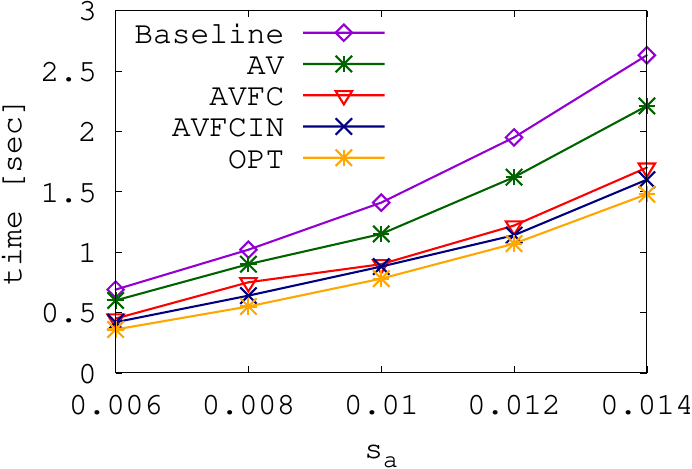}
    }\!\!
  \subfigure[Flights Network]{
   \label{fig:exp:saflights}
   \includegraphics[width=0.3\textwidth]{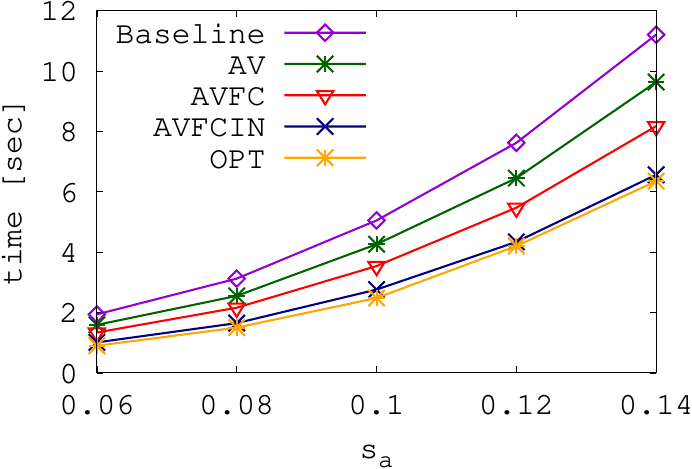}
    }
  \vspace{-0.1in}
  \caption{Pattern enumeration runtime, $s_r=0.5$, varying $s_a$}
  \label{fig:exp:time}
\end{figure*}

\begin{figure*}[t!]
  \centering
 \subfigure[Taxi Network]{
   \label{fig:exp:taxitime}
    \includegraphics[width=0.3\textwidth]{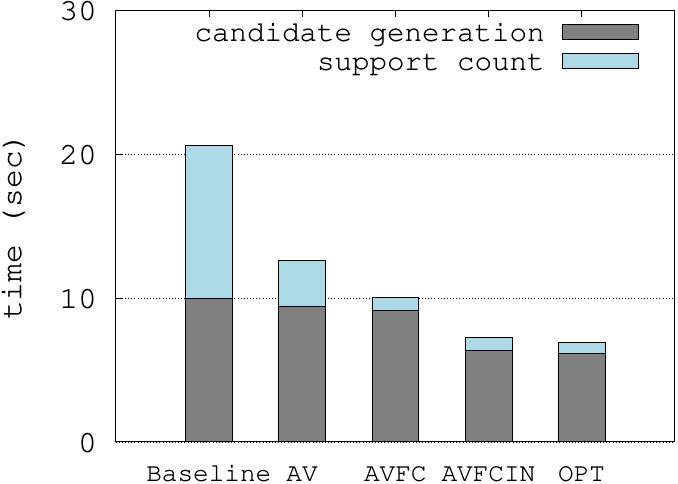}
    }\!\!
  \subfigure[Metro Network]{
   \label{fig:exp:mtrtime}
    \includegraphics[width=0.3\textwidth]{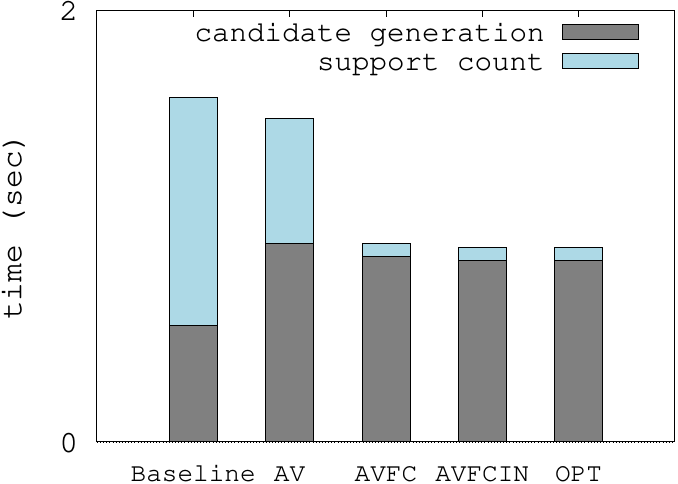}
    }\!\!
  \subfigure[Flights Network]{
   \label{fig:exp:flightstime}
   \includegraphics[width=0.3\textwidth]{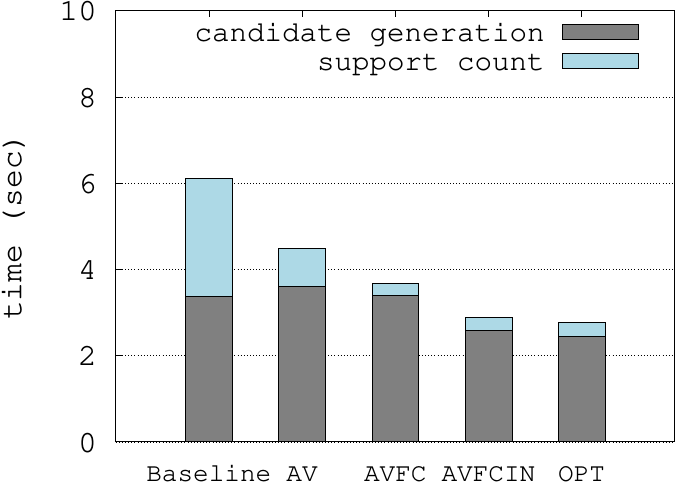}
    }
  \vspace{-0.1in}
  \caption{Pattern enumeration cost breakdown, $s_r=0.5$, default $s_a$}
  \label{fig:exp:timebreakdown}
\end{figure*}

\begin{figure*}[t!]
  \centering
\vspace{-2mm}
 \subfigure[Taxi Network]{
   \label{fig:exp:srtaxi}
    \includegraphics[width=0.3\textwidth]{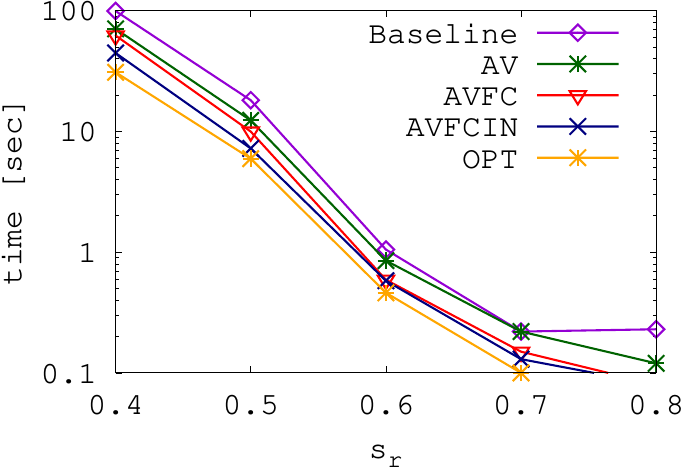}
    }\!\!
  \subfigure[Metro Network]{
   \label{fig:exp:srmtr}
    \includegraphics[width=0.3\textwidth]{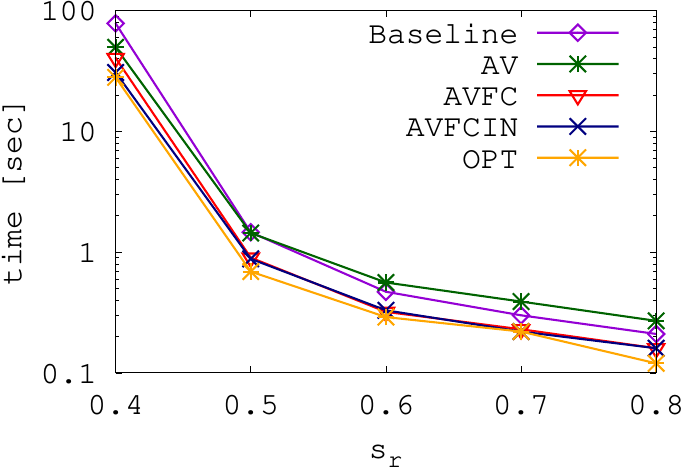}
    }\!\!
  \subfigure[Flights Network]{
   \label{fig:exp:srflights}
   \includegraphics[width=0.3\textwidth]{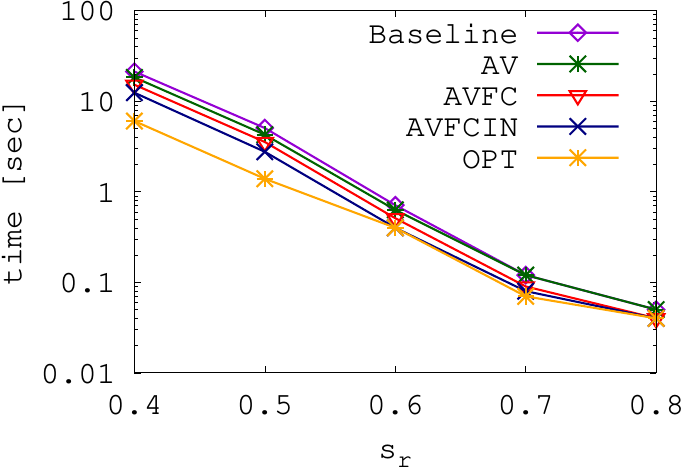}
    }
  \vspace{-0.1in}
  \caption{Pattern enumeration runtime, default $s_a$, varying $s_r$}
  \label{fig:exp:sr}
\end{figure*}

\begin{figure*}[t!]
  \centering
\vspace{-2mm}
 \subfigure[Taxi Network]{
   \label{fig:exp:taxipatterns}
    \includegraphics[width=0.3\textwidth]{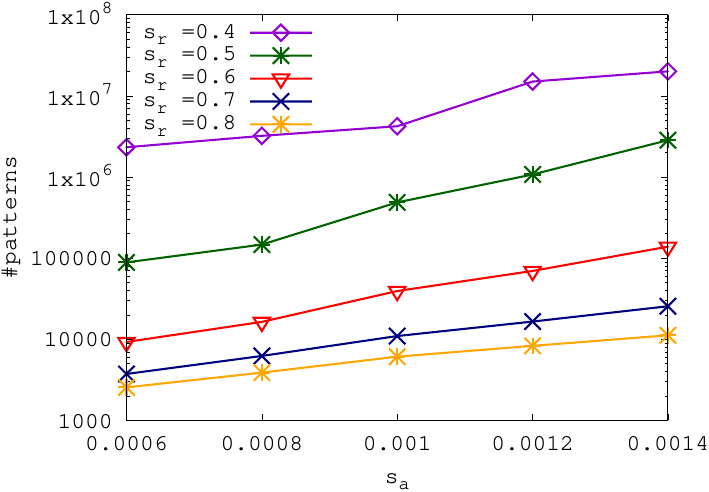}
    }\!\!
  \subfigure[Metro Network]{
   \label{fig:exp:mtrpatterns}
    \includegraphics[width=0.3\textwidth]{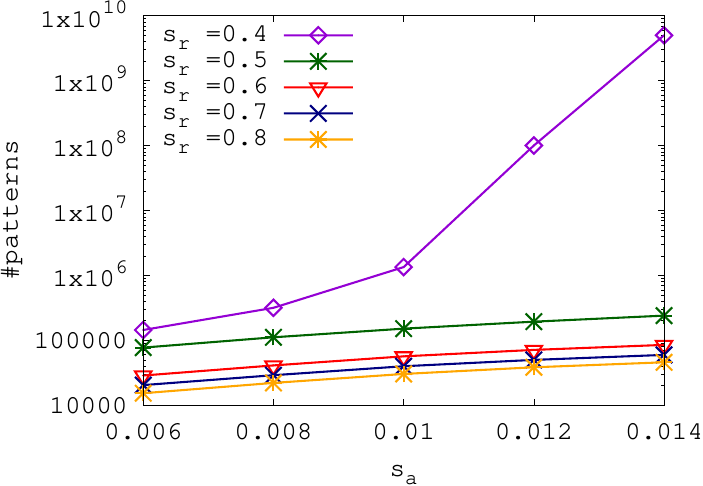}
    }\!\!
  \subfigure[Flights Network]{
   \label{fig:exp:flightspatterns}
   \includegraphics[width=0.3\textwidth]{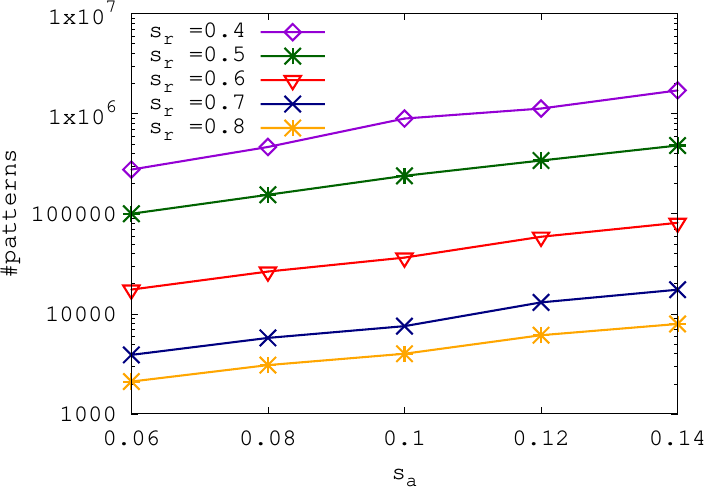}
    }
  \vspace{-0.1in}
  \caption{Number of patterns for different values of $s_a$ and $s_r$}
  \label{fix:exp:patterns}
\end{figure*}


\begin{figure*}[t!]
  \centering
\vspace{-2mm}
 \subfigure[Taxi Network]{
   \label{fig:exp:srctaxi}
    \includegraphics[width=0.3\textwidth]{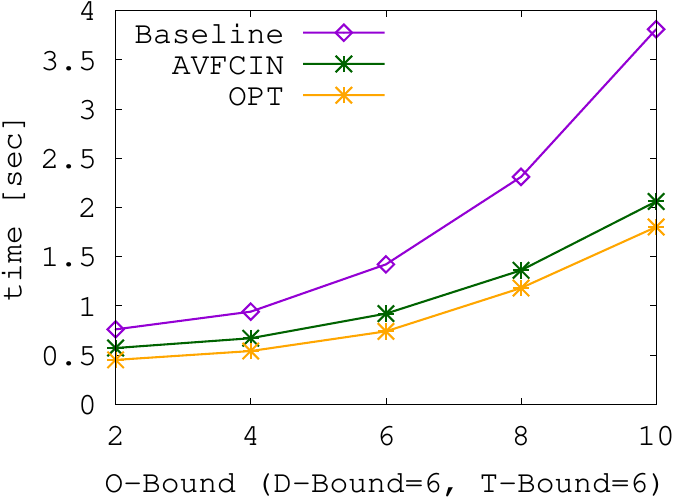}
    }\!\!
  \subfigure[Metro Network]{
   \label{fig:exp:srcmtr}
    \includegraphics[width=0.3\textwidth]{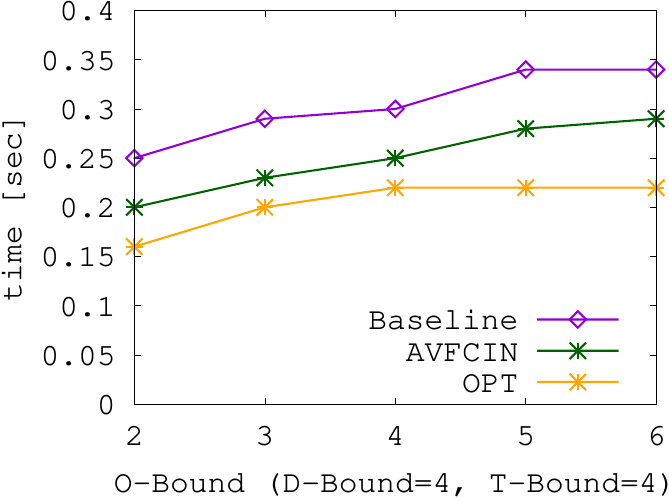}
    }\!\!
\subfigure[Flights Network]{
   \label{fig:exp:srcflights}
   \includegraphics[width=0.3\textwidth]{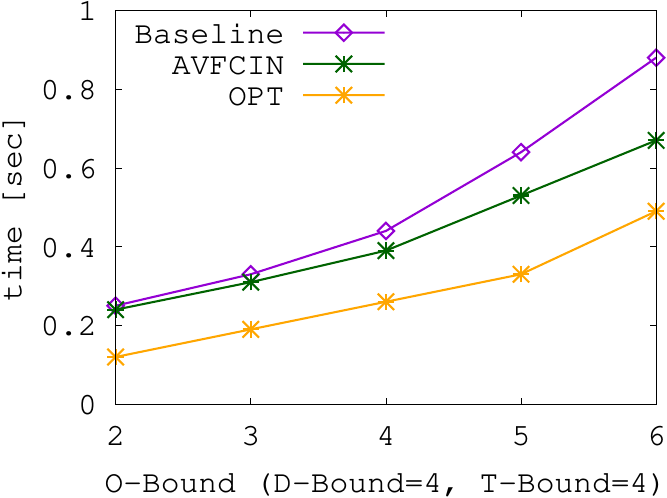}
    }
  \vspace{-0.1in}
  \caption{Bounded pattern enumeration runtime, default $s_a$, $s_r$,
    varying origin bound}
  \label{fig:exp:src}
\end{figure*}

\begin{figure*}[t!]
  \centering
\vspace{-2mm}
 \subfigure[Taxi Network]{
   \label{fig:exp:desttaxi}
    \includegraphics[width=0.3\textwidth]{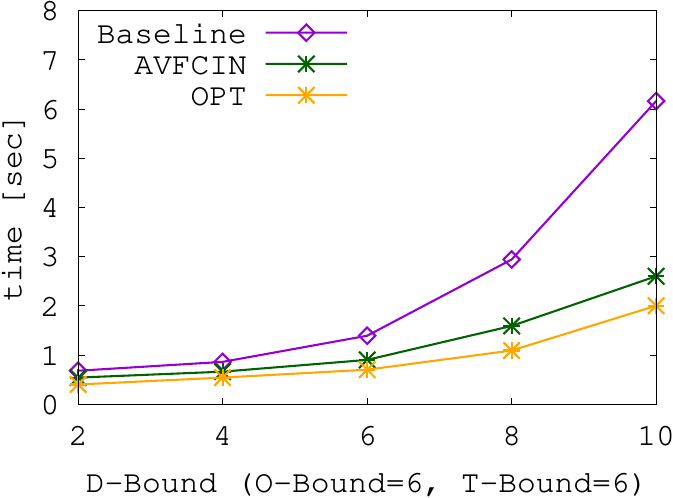}
    }\!\!
  \subfigure[Metro Network]{
   \label{fig:exp:destmtr}
    \includegraphics[width=0.3\textwidth]{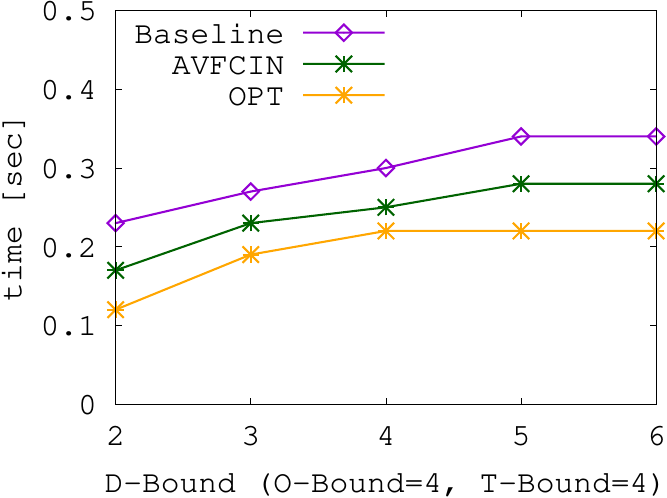}
    }\!\!
  \subfigure[Flights Network]{
   \label{fig:exp:destflights}
   \includegraphics[width=0.3\textwidth]{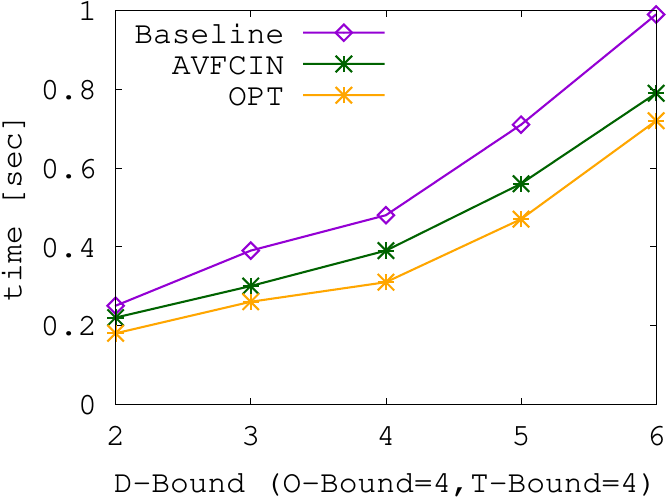}
    }
  \vspace{-0.1in}
  \caption{Bounded pattern enumeration runtime, default $s_a$, $s_r$,
    varying destination bound}
  \label{fig:exp:dest}
\end{figure*}

\begin{figure*}[t!]
  \centering
\vspace{-2mm}
 \subfigure[Taxi Network]{
   \label{fig:exp:timetaxi}
    \includegraphics[width=0.3\textwidth]{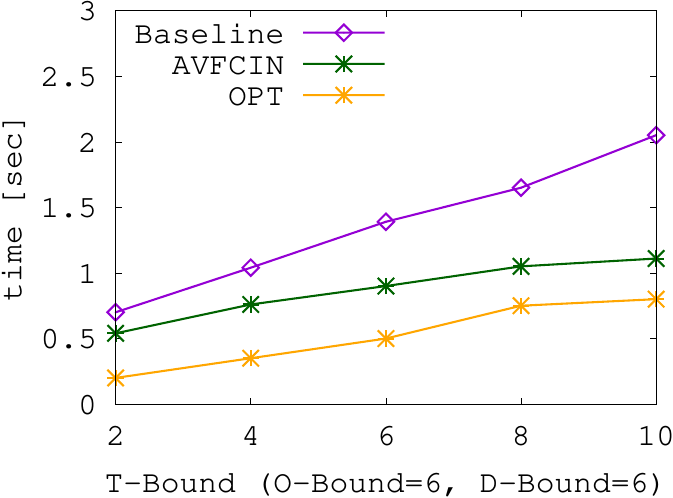}
    }\!\!
  \subfigure[Metro Network]{
   \label{fig:exp:timemtr}
    \includegraphics[width=0.3\textwidth]{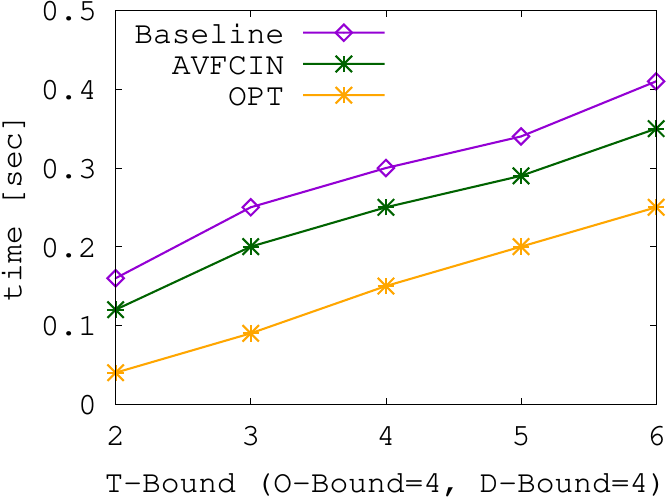}
    }\!\!
  \subfigure[Flights Network]{
   \label{fig:exp:timeflights}
   \includegraphics[width=0.3\textwidth]{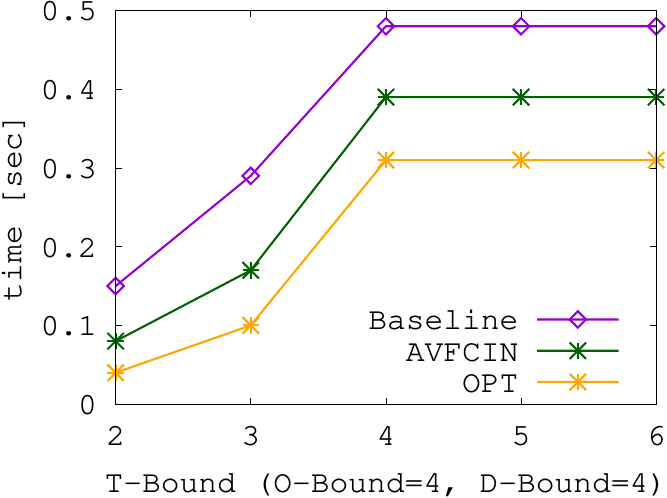}
    }
  \vspace{-0.1in}
  \caption{Bounded pattern enumeration runtime, default $s_a$, $s_r$,
    varying timeslot bound}
  \label{fig:exp:timec}
\end{figure*}


\begin{figure*}[t!]
  \centering
\vspace{-2mm}
 \subfigure[Taxi Network]{
   \label{fig:exp:taxilevel}
    \includegraphics[width=0.3\textwidth]{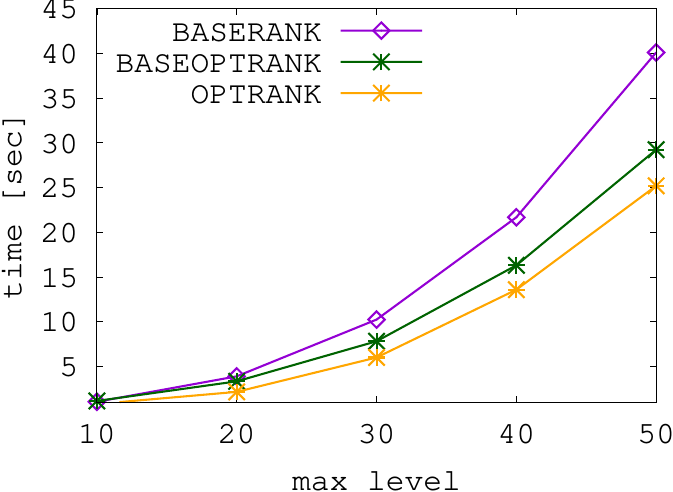}
    }\!\!
  \subfigure[Metro Network]{
   \label{fig:exp:mtrlevel}
    \includegraphics[width=0.3\textwidth]{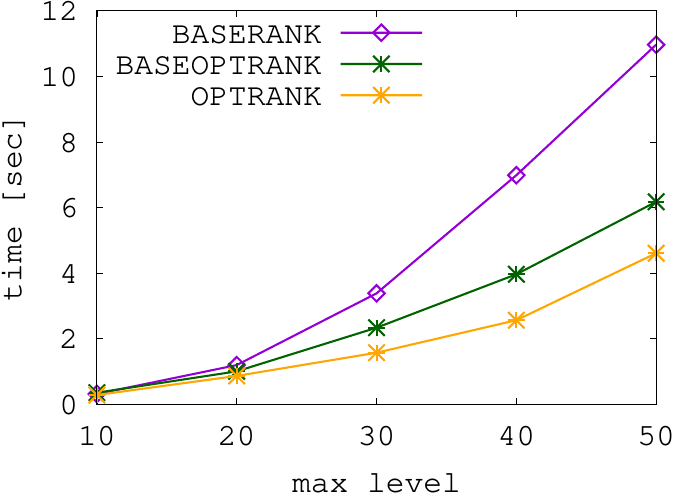}
    }\!\!
  \subfigure[Flights Network]{
   \label{fig:exp:flightslevel}
   \includegraphics[width=0.3\textwidth]{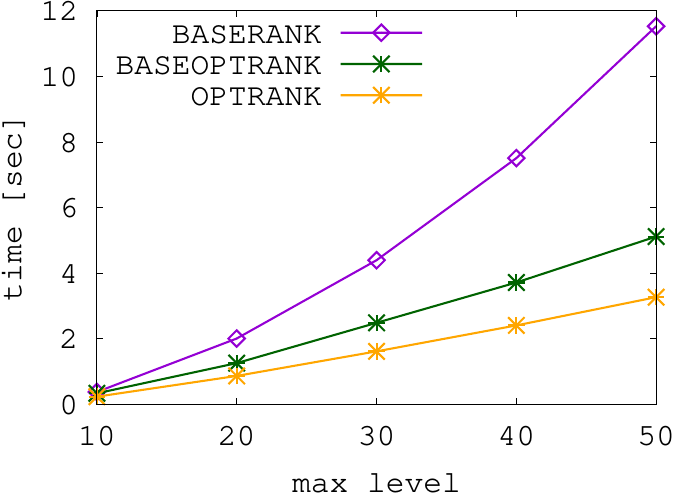}
    }
  \vspace{-0.1in}
  \caption{Rank-based pattern enumeration, $s_a=0.1$, $k=3000$, varying $maxl$}
  \label{fig:exp:level}
\end{figure*}

\begin{figure*}[t!]
  \centering
\vspace{-2mm}
 \subfigure[Taxi Network]{
   \label{fig:exp:taxik}
    \includegraphics[width=0.3\textwidth]{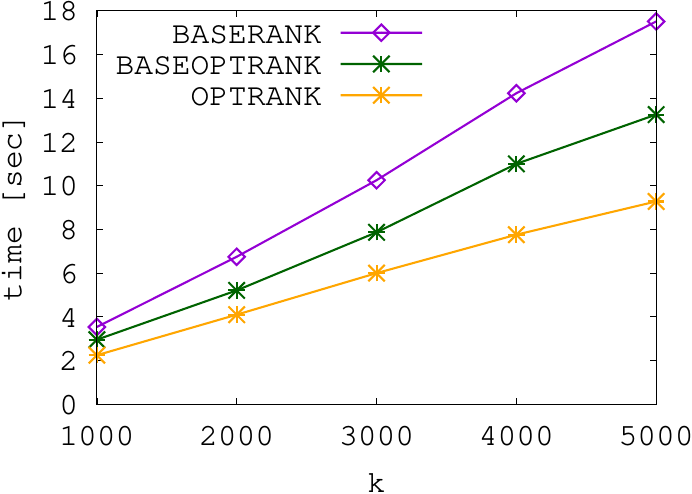}
    }\!\!
  \subfigure[Metro Network]{
   \label{fig:exp:mtrk}
    \includegraphics[width=0.3\textwidth]{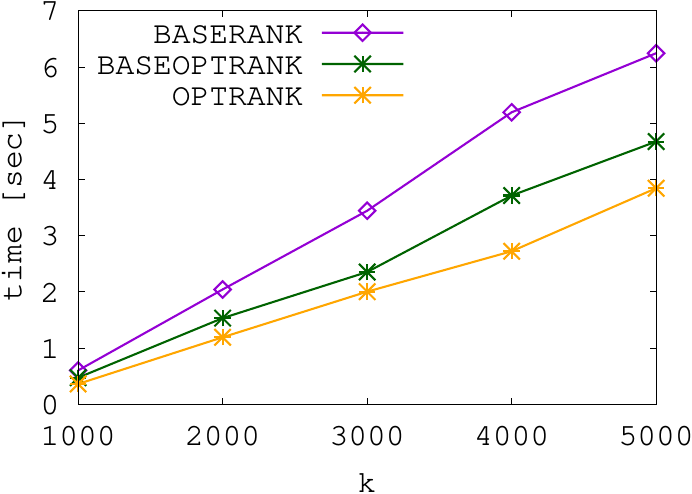}
    }\!\!
  \subfigure[Flights Network]{
   \label{fig:exp:flightsk}
   \includegraphics[width=0.3\textwidth]{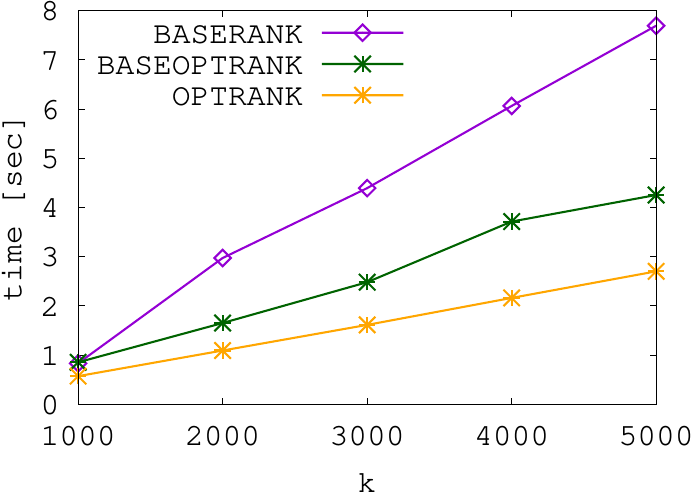}
    }
  \vspace{-0.1in}
  \caption{Rank-based pattern enumeration, $s_a=0.1$, $maxl=30$, varying $k$}
  \label{fig:exp:k}
\end{figure*}

\subsection{Pattern enumeration}
We start by evaluating the performance of our
baseline pattern enumeration algorithm, described in Section
\ref{sec:baseline}, and its optimizations, described in Section \ref{sec:opt}.
Specifically, we compare the performance of the following methods:

  \begin{itemize}
  \item Algorithm \ref{algo:baseline}, denoted by Baseline.
  \item Algorithm \ref{algo:baseline}  with the avoid recounting $P'$
    optimization, denoted by AV. 
    \item Algorithm \ref{algo:baseline}  with the avoid recounting $P'$
      and fast check for zero support of $P'$ optimizations, denoted by AVFC.
    \item Algorithm \ref{algo:baseline}  with the avoid recounting $P'$,
      fast check, and improved neighborhood optimizations, denoted by AVFCIN.
     \item Algorithm \ref{algo:baseline}  with all four optimizations,
       denoted by OPT.
    \end{itemize}

Figure \ref{fig:exp:time} shows the costs of all tested methods on the three datasets
for various values of $s_a$ (default $s_a=0.001$ for Taxi, 
$s_a=0.01$ for Metro, and $s_a=0.1$ for Flights), while keeping $s_r$ fixed to 0.5. Observe
that the optimizations pay off, since the initial cost of the baseline
approach drops to about 50\% of the initial cost. When comparing
between the different optimizations, we observe that the ones that
have the biggest impact are the $P'$ counting avoidance and the
improved neighborhood computation. The savings by the prefix sum
optimization are not impressive, because the
other optimizations already reduce a lot the number of candidates for which
exact counting is required.

This assertion is confirmed by the cost-breakdown experiment shown in
Figure \ref{fig:exp:timebreakdown}, where for the default values of $s_a$
and $s_r$, we show the fraction of the cost that goes to candidate
pattern generation and support counting. Note that the baseline
approach spends most of the time in pattern counting, as the candidate
generation process is quite simple. On the other hand, the optimized
versions of the algorithm trade off time for pattern generation (spent
on bookkeeping all generated triples at each level, bookkeeping OD pairs
with at least one trip, etc.) to reduce the time spent on support
counting. Note that the ratio of the time spent on support counting is
eventually minimized. When comparing between the different versions,
we observe that the candidate generation time drops as more
optimizations are employed (e.g., fast check for zero support).
Since finding all patterns at level $\ell$ requires
considering all possible extensions of patterns at level $\ell-1$, we
note that there is little room for further reducing the cost of ODT pattern
enumeration; in this respect OPT is the best approach that one can
apply if the goal is to find all ODT patterns.

Figure \ref{fig:exp:sr} shows the runtime cost of pattern enumeration
for different values of $s_r$, by keeping $s_a$ to its default
value. Observe that the cost explodes for values of $s_r$ smaller than
$0.5$. The reason is that small $s_r$ values make it easy for 
triples at each level to be characterized as patterns, which, in turn,
greatly increases the number of candidates and patterns at the next
level. On the other hand, for $s_r\ge 0.5$ at least half of the atomic
triples in a candidate must be atomic patterns, which restricts
the number of candidates and patterns at all levels.

The next experiment proves the pattern explosion for small values of
$s_r$.
The high cost of pattern enumeration stems from the fact that a very
large number of patterns are found at each level, which, in turn, all
have to be minimally generalized due to the weak monotonicity property
of Definition \ref{def:pattern}. Figure \ref{fix:exp:patterns}
shows
the numbers of enumerated patterns for different values of $s_a$ and
$s_r$. As the number of patterns grow, so does the essential cost of
candidate generation, which becomes the dominant cost factor.
From Figure \ref{fix:exp:patterns}, we observe that the number of
enumerated patterns is very sensitive to $s_r$. Specifically, for
values of $s_r$ smaller than $0.5$ the number of patterns explode. On
the other hand, the sensitivity to $s_a$ is relatively low. Still,
even for the default values of $s_a$ (0.001 for Taxi and 0.01 for Metro
and Flights) there are thousands or even millions of qualifying
patterns. Such huge numbers necessitate the use of constraints or
ranking in order to limit the number of patterns, focusing on the most
important ones.  

\subsection{Bounded patterns}


As discussed in Section \ref{sec:sizebounded}, one way to limit the
number of patterns is to bound the number of atomic regions and/or
atomic timeslots in them. In the next experiment, we study the effect
of such pattern size constraints to the runtime of algorithms
Baseline, AVFCIN, and OPT.
We run experiments by setting $s_a$ and $s_r$ to their default values.
In each experiment, we set a fixed upper bound to the sizes of two of O, D,
and T, and vary the bound of one. Hence, in Figure \ref{fig:exp:src},
we keep the upper size bounds of D and T fixed and we vary the upper size bound of
O; in Figure \ref{fig:exp:dest},
we keep the upper size bounds of O and T fixed and we vary the upper size bound of
D; in Figure \ref{fig:exp:timec},
we keep the upper size bounds of O and D fixed and we vary the upper
size bound of T. In general, the cost increases as
one bound increases, which is as expected, because the number of
patterns and generated candidates increases as well. On certain
datasets (e.g., Metro), the cost growth is slow when the bound of O or D
is increased; this is due to the fact that the number of patterns at
low levels is already quite small and the generated patterns start to
decrease as we change levels, so the bound increase does not affect
the cost significantly. On the other
hand, when the bound of T increases (Figure \ref{fig:exp:timec}),
there is a stable increase of time in all datasets. This is due to the
fact that the number of atomic timeslots is significantly small and
neighboring timeslots are highly correlated in terms of flow.
When comparing the costs of Baseline, AVFCIN, and OPT, we observe that
OPT maintains a significant performance advantage for different bound
values, especially on Metro.


\subsection{Rank-based patterns}

We now evaluate the performance of rank-based
pattern enumeration, described in Section \ref{sec:rank}.
We compare three algorithms.
The first one is the baseline approach
described in Section \ref{sec:baselinerank}, without the pattern
enumeration optimizations described in Section \ref{sec:opt}.
The second one is the baseline approach of Section
\ref{sec:baselinerank}
with the pattern
enumeration optimizations described in Section \ref{sec:opt}.
The third approach is the optimized algorithm for rank-based patterns
described in Section \ref{sec:optrank}. The three approaches are
denoted by BASERANK, BASEOPTRANK, and OPTRANK, respectively.

Figure \ref{fig:exp:level}
 shows the runtime cost of the three algorithms for
$s_a=0.1$ and $k=3000$ patterns per level, as a function of
the maximum level $maxl$ of patterns that we generate and enumerate. Recall
that the top-$k$ patterns selected per level may generate numerous
triples at the next level and there is no $s_r$ threshold to reduce
them, so the number of levels can become too large. We use $maxl$ as a
parameter for limiting the sizes of patterns.
As shown in the figure, OPTRANK maintains a large advantage over the
other approaches which do not take advantage of the pruning conditions
and the ranking of generated triples.
Figure \ref{fig:exp:k} shows the runtime cost of the algorithms for $s_a=0.1$
and various
values of $k$, after setting $maxl=30$. The advantage of  OPTRANK over
the other algorithms is not affected by $k$.
Overall, despite the fact that a very high value of $s_a$ is used, due
to the fact that the number of patterns per level is limited by $k$,
all algorithms are scalable, making pattern enumeration practical,
even in cases where the number of possible ODT combinations is huge.

\subsection{Use cases}
Finally, we explored the use of ODT patterns in real applications.
We restricted the origin and time dimensions,
according to Section \ref{sec:restricted},
and identified the most popular
(generalized) destinations.

Table \ref{table:usecase2} shows some of these 
patterns in the Taxi dataset.
We first restricted O to be GreenPoint, Brooklyn and T to peak hour morning
timeslots. This gave us as most popular destinations, extended region
Williamsburg East and South and extended region
\{Williamsburg E, Williamsburg S,
Williamsburg NS, Williamsburg SS\}.
In afternoon peak hours people from a central region in Manhattan
(Midtown South) tend to move to neighboring central regions (MidTown
Centre, MidTown East, Times Square, Murray Hill).
Overall, based on our study, most people move within their borough to
relatively near destinations (possibly due to high taxi fares).

\begin{table}[ht]
\caption{Use case - Taxi Dataset}
\vspace{-0.2cm}
\centering
 \scriptsize
\begin{tabular}{|@{~}c@{~}|@{~}c@{~}|@{~}c@{~}|}
\hline
Origin &Timeslots & popular destinations\\
\hline
  \hline
GreenPoint&[8:30-9:30]& WilliamsbE, WilliamsbS\\
GreenPoint&[8:30-9:30]& WilliamsbE,WilliamsbS,WilliamsbNS,WilliamsbSS\\
\hline
Midtown South&[17:30-18:30]&MidTownCentre, MidTownEast\\
Midtown South&[17:30-18:30]&MidTownCentre,MidTownEast,TimesSquare,MurrayHill\\[0.2ex]  
\hline
\end{tabular}
\label{table:usecase2}
\end{table}